\documentclass{article}
\usepackage[utf8]{inputenc}
\usepackage{mathrsfs}
\usepackage{fullpage}
\usepackage{amsfonts}
\usepackage{graphicx}
\usepackage{amsmath}
\usepackage{amssymb}
\usepackage{float}
\usepackage{subfig}
\usepackage{rotating}
\usepackage{xcolor}
\usepackage{natbib}
\usepackage{epstopdf}
\usepackage{enumerate}
\usepackage{url}
\usepackage[space]{grffile}
\usepackage{latexsym}
\usepackage{textcomp}
\usepackage{longtable}
\usepackage{graphicx} 
\usepackage{tabulary}
\usepackage{booktabs,array,multirow}
\usepackage{epstopdf}
\usepackage{epsfig,fancyheadings}
\usepackage{sectsty, secdot}
\usepackage{amsthm}
\usepackage{bm,bbm}
\usepackage{setspace}
\usepackage{geometry,algorithm,algpseudocode}
\usepackage[colorlinks,citecolor=blue,urlcolor=blue,linkcolor = blue]{hyperref}

\newtheorem{theorem}{Theorem}
\newtheorem{lemma}{Lemma}
\newtheorem{remark}{Remark}

\def\E{\mathrm{E}}
\def\Pr{\mathrm{P}}
\def\var{\mathrm{var}}

\def\cov{\mathrm{cov}}

\def\cp{\mathop{\rightarrow}\limits^{p}}
\def\cd{\mathop{\rightarrow}\limits^{d}}
\def\mR{\mathbb{R}}
\def\mS{\mathbb{S}}
\def\indicate{\mathbbm{1}}
\newcommand*{\dif}{\mathop{}\!\mathrm{d}}

\def\I{{\bf I}}

\def\O{{\bf O}}
\def\P{{\bf P}}
\def\Q{{\bf Q}}

\def\U{{\bf U}}
\def\V{{\bf V}}

\def\X{{\bf X}}

\def\Z{{\bf Z}}

\def\e{{\bm e}}

\def\s{{\bm s}}
\def\u{{\bf u}}
\def\x{{\bm x}}

\def\bms{{\bf \Sigma}}

\def\bmu{{\bm \mu}}

\def\bzero{{\bm 0}}
\def\bone{{\bm 1}}

\allowdisplaybreaks[4]

\title{Adaptive L-tests for high dimensional independence}
\author{Ping Zhao\\
School of Mathematical Sciences, Tiangong University\\
Huifang Ma\\
School of Statistics and Data Science, Nankai University}
\date{\today}

\begin{document}

\maketitle

\begin{abstract}
Testing mutual independence among multiple random variables is a fundamental problem in statistics, with wide applications in genomics, finance, and neuroscience. 
In this paper, we propose a new class of tests for high-dimensional mutual independence based on $L$-statistics. We establish the asymptotic distribution of the proposed test when the order parameter $k$ is fixed, and prove asymptotic normality when $k$ diverges with the dimension. Moreover, we show the asymptotic independence of the fixed-$k$ and diverging-$k$ statistics, enabling their combination through the Cauchy method. 
The resulting adaptive test is both theoretically justified and practically powerful across a wide range of alternatives. Simulation studies  demonstrate the advantages of our method. 

\end{abstract}

\section{Introduction}

Testing the independence between random variables is one of the most fundamental problems in statistics. Historically, a variety of methods have been developed to test the independence between two variables, including correlation-based approaches such as Pearson's correlation \citep{pearson1895vii}, Kendall's tau \citep{kendall1938new}, Spearman's rho \citep{spearman1987proof}, and Hoeffding's $D$ test \citep{hoeffding1948non}. These classical methods provide powerful tools in the bivariate setting. When extending to the multivariate case, a natural formulation of the mutual independence testing problem arises: given a $p$-dimensional random vector, we aim to test whether all its components are mutually independent. 
Under the assumption of multivariate normality, this problem reduces to testing whether the correlation matrix is the identity matrix. So the classic likelihood ratio test could be used \citep{anderson2003introduction}. 

In recent years, fueled by the increasing prevalence of high-dimensional data in applications such as genomics, finance, and neuroscience, there has been growing interest in developing mutual independence tests that remain valid in high-dimensional settings.  The problem of testing mutual independence is not only of theoretical interest but also has strong practical motivations. In financial econometrics, detecting cross-sectional independence among asset returns is essential for risk management and portfolio allocation. In genomics, independence testing helps to uncover gene regulatory networks. In neuroscience, testing whether brain regions exhibit synchronous activity relies on accurate assessment of independence. Therefore, developing statistically sound and computationally efficient tests for high-dimensional mutual independence remains an important and timely research direction.

Broadly speaking, existing high-dimensional mutual independence tests can be divided into three categories.  
The first class is based on \emph{sum-type statistics}, which aggregate the squared pairwise dependence measures across all variable pairs \citep{schott2005testing,srivastava2005some,srivastava2006some,srivastava2011some,gao2017high,li2023test}. 
Such statistics typically converge to a normal distribution under the null, and are particularly powerful under dense alternatives where many pairs exhibit weak dependence.
Another sum-type test is extending the classic likelihood ratio test to high dimensional settings, such as \cite{jiang2013central,jiang2019determinant,heiny2024log,hu2023limiting,gao2017high}. 
The second class is based on \emph{max-type statistics}, which take the maximum of pairwise dependence measures and often converge to a Gumbel distribution \citep{jiang2004asymptotic,liu2008asymptotic,cai2011limiting}. 
These tests are well-suited to sparse alternatives where only a small fraction of pairs are dependent .  
The third class combines the strengths of both approaches by constructing \emph{adaptive tests}. 
The general idea is to first establish the asymptotic independence between sum-type and max-type statistics, and then combine their $p$-values through methods such as the Cauchy combination test, yielding tests that are powerful across a wide range of sparsity levels \citep{li2015joint,feng2022max,wang2024rank}.  

This paper advances the third line of research on high-dimensional mutual independence testing by proposing a novel framework grounded in $L$-statistics. Although $L$-statistics have been extensively studied in the context of robust estimation and nonparametric inference \citep{serfling1980approximation}, their application to high-dimensional dependence testing remains largely unexplored. Building upon recent developments in adaptive $L$-tests for high-dimensional mean testing \citep{ma2024adaptive}, we introduce an $L$-type test statistic indexed by $k$, representing the number of largest pairwise dependence measures incorporated into the test.

First, we establish the limiting distribution of the proposed statistic under fixed $k$, showing that it is particularly powerful against sparse alternatives. Second, when $k$ diverges with the dimension $p$, we prove the asymptotic normality of the statistic, which makes it well-suited for dense alternatives. Third, we demonstrate that the fixed-$k$ and diverging-$k$ statistics are asymptotically independent, thereby enabling a principled combination via the Cauchy combination method \citep{liu2020cauchy}. Fourth, we address the fundamental challenge that the true sparsity level of alternatives—namely, the number of nonzero correlations $s$—is unknown in practice. While \cite{ma2024adaptive} showed that the optimal choice of $k$ coincides with $s$, this information is unavailable to practitioners. To overcome this limitation, we propose a data-adaptive approach that evaluates a range of $k$ values and integrates the resulting tests, thereby achieving robustness and efficiency across different sparsity regimes.

The resulting adaptive $L$-test inherits the strengths of both fixed-$k$ and diverging-$k$ statistics, offering uniformly superior power against a wide spectrum of alternatives, from extremely sparse to highly dense. Extensive simulation studies confirm the theoretical advantages, and a real data application further highlights its practical relevance. Taken together, our results establish $L$-statistics as a powerful and flexible tool for high-dimensional independence testing, complementing and extending existing sum-type and max-type methodologies.

The remainder of this paper is organized as follows.  
In Section 2, we introduce the proposed test statistics and present their asymptotic properties.  
Section 3 reports simulation results that illustrate the advantages of our methods under different sparsity regimes.  
Section 4 concludes with some discussions.  
All technical proofs and additional results are provided in the supplementary material.

\section{L-tests}
In this paper, we consider a test for the independence of the variables comprising the $p\times 1$ vector $\X=(X_1,\ldots,X_p)^\top$, i.e., testing the global mutual independence hypothesis
\begin{align}\label{h0}
H_0:\; X_1,\ldots,X_p \;\text{are mutually independent random variables.}
\end{align}
Suppose we observe $n$ independent and identically distributed samples $\X_k=(X_{k1},\ldots,X_{kp})^{\top}$, $1\le k\le n$ from the following independent component (IC) model \citep{bai2004clt},
\begin{align}\label{model}
    \X_k=\bmu+\bms^{1/2}\Z_k,~~k=1,\ldots,n,
\end{align}
where $\bmu=(\mu_1,\ldots,\mu_p)^\top$ and $\bms=(\sigma_{ij})_{1\le i,j\le p}$ denote the unknown mean vector and covariance matrix, respectively, and $\Z_k=(Z_{k1},\ldots,Z_{kp})^{\top}$ is a vector of independent random variables with zero means and unit variances. If $\rho_{ij}=\sigma_{ij}/\sqrt{\sigma_{ii}\sigma_{jj}}$ is the correlation coefficient between variables $i$ and $j$, then hypothesis \eqref{h0} can be written as $H_0: \rho_{ij}=0,~\forall~i<j$. The sample correlation coefficient between variables $i$ and $j$ is given by
\begin{align*}
\hat{\rho}_{ij}=\frac{\hat{\sigma}_{ij}}{\sqrt{\hat{\sigma}_{ii} \hat{\sigma}_{jj}}},
\end{align*}
where $\bar{X}_{.i}=\frac{1}{n}\sum_{k=1}^nX_{ki}$ is the sample mean and $\hat{\sigma}_{ij}=\frac{1}{n}\sum_{k=1}^n (X_{ki}-\bar{X}_{.i})(X_{kj}-\bar{X}_{.j})$ is the sample covariance.

\cite{schott2005testing} proposed a sum-type test statistic based on the squared correlations, i.e.,
\begin{align*}
T_{SC}=\frac{\sum_{j=2}^p \sum_{i=1}^{j-1} n\hat{\rho}_{ij}^2-p^*}{\sqrt{2p^*(n-1)/(n+2)}},
\end{align*}
where $p^*=p(p-1)/2$ and established its asymptotic normality under the null. This type of quadratic (or $L_2$-norm) statistic has good power performance under dense alternatives, where many weak correlations exist simultaneously.

To improve sensitivity to sparse signals, \cite{li2023test} proposed a different sum-type statistic based on the $L_4$-norm of the sample covariances,
\begin{align*}
T_{LX}=\sum_{i=2}^p \sum_{j=1}^{i-1} \hat{\sigma}_{ij}^4.
\end{align*}
Unlike $T_{SC}$, this test does not rely on the sample correlations and hence avoids scale-invariance issues. Their results demonstrate that $T_{LX}$ is more powerful than $T_{SC}$ under moderately sparse alternatives. However, when the alternative is extremely sparse, max-type statistics are more suitable.

For this purpose, \cite{jiang2004asymptotic} considered the maximum of the squared correlations,
\begin{align*}
T_J=\max_{1\le i<j\le p} n\hat{\rho}_{ij}^2-4\log p+\log \log p,
\end{align*}
and established that its limiting distribution is Gumbel under the null hypothesis. Max-type tests such as $T_J$ are particularly effective for detecting alternatives with only a few strong correlations.

In practice, the sparsity level of the alternative hypothesis is unknown, and it is not clear whether a dense- or sparse-type test is more appropriate. To address this challenge, \cite{li2015joint} and \cite{feng2022max} proved the asymptotic independence between $T_{SC}$ and $T_J$, and proposed the following combined test statistic:
\begin{align*}
T_F=\min\{p_{SC},p_{J}\},
\end{align*}
where $p_{SC}$ and $p_J$ are the $p$-values of $T_{SC}$ and $T_J$, respectively. This adaptive procedure achieves robustness against unknown sparsity levels by combining dense- and sparse-type tests. However, $T_F$ only exploits the global sum and the extreme maximum of the correlations. When the alternative sparsity lies in an intermediate regime, both $T_{SC}$ and $T_J$ may lose efficiency, leading to suboptimal power.

To overcome this limitation, we propose a more flexible approach based on $L$-statistics. Specifically, let $p^*=p(p-1)/2$ denote the total number of distinct correlation pairs, and order the squared sample correlations in descending order as $\hat{\rho}^2_{(1)} \le \hat{\rho}^2_{(2)} \le \cdots \le \hat{\rho}^2_{(p^*)}$. For a given parameter $k$, we define the $L$-type test statistic
\begin{align}\label{L-statistic}
T_k=\sum_{i=1}^k n\hat{\rho}^2_{(p^*-i+1)}.
\end{align}
This statistic aggregates the top $k$ largest correlation signals and therefore interpolates between max-type tests ($k=p^*$) \citep{jiang2004asymptotic} and sum-type tests ($k=1$) \citep{schott2005testing}. By varying $k$, the $L$-test provides a natural framework for adapting to a wide range of sparsity regimes.

In the following sections, we establish the asymptotic distributions of $T_k$ under different regimes. When $k$ is fixed, we show that $T_k$ is particularly powerful for sparse alternatives. When $k$ diverges with $p$, we prove its asymptotic normality, which favors dense alternatives. Finally, we show that the two regimes are asymptotically independent, allowing us to construct an adaptive procedure that combines them efficiently via the Cauchy combination method.

\subsection{Fixed parameter}\label{fixp}

We now state our first main result, which is about the limiting distribution of each $n\hat{\rho}^2_{(j)}$ and the joint limiting distribution of all $n\hat{\rho}^2_{(j)}$'s.

\begin{theorem}\label{asymptotic distribution of ordered statistics}
	Under the IC model given in \eqref{model}, suppose $\E\{\exp(t_0|Z_{11}|^{a})\}<\infty$ for some $0<a\le 2$ and $t_0>0$. Assume $p=p(n)\to \infty$ and $\log p=o(n^\beta)$ as $n\to\infty$, where $\beta=a/(4+a)$. Then as $\min(n,p)\to\infty$, we have
	
	\noindent(\romannumeral1) for all integer $1\le s\le p^*$ and $x\in\mR$,
	\begin{align*}
	\Pr_{H_0}\left(n\hat{\rho}^2_{(p^*+1-s)}-b_p\le x \right)\to \Lambda(x)\sum_{i=0}^{s-1}\frac{\{\log \Lambda^{-1}(x)\}^i}{i!},
	\end{align*}
	\noindent(\romannumeral2) for all integer $2\le l\le p^*$ and $x_1\ge \cdots\ge x_l\in\mR$,
	\begin{align*}
	&\Pr_{H_0}\left(\bigcap_{s=1}^l\left(n\hat{\rho}^2_{(p^*+1-s)}-b_p \leq x_s\right)\right)\\
	\to&\Lambda\left(x_l\right) \sum_{\sum_{i=2}^s l_i \leq s-1, s=2, \ldots, l} \prod_{i=2}^l \frac{\{\log \Lambda^{-1}(x_i)-\log \Lambda^{-1}(x_{i-1})\}^{l_i}}{l_i!},
	\end{align*}
	where $b_p=4\log p-\log (\log p)$ and $\Lambda(x)=\exp\{-(8\pi)^{-1/2}\exp (-x/2)\}$.
\end{theorem}
According to Theorem \ref{asymptotic distribution of ordered statistics}, through convolution, for fixed $k$ the asymptotic null distribution of $T_k$ exists. Specially, $\Pr_{H_0}(T_1-b_p\le x)\to \Lambda(x)$. When $k>1$, the asymptotic null distribution of $T_k$ is complex. So we adapt the permutation method to get the null distribution of $T_k$.

To approximate the null distribution of the proposed test statistic $T_k$, we adopt a permutation-based bootstrap procedure. Specifically, we generate bootstrap samples 
\begin{align*}
    \X_k^* = (X_{k1}^*,\ldots,X_{kp}^*)^{\top}, \quad k=1,\ldots,n,
\end{align*}
where, for each coordinate $i$, $\{X_{ki}^*\}_{k=1}^n$ is obtained by randomly permuting the original sample $\{X_{ki}\}_{k=1}^n$ without replacement. This procedure effectively breaks any potential dependence among variables while preserving the marginal distributions of each coordinate, thereby mimicking the null hypothesis of mutual independence.

For each bootstrap sample, we construct the corresponding test statistic $T_k^*$ using the same procedure as in the original statistic $T_k$. Repeating this process independently $B$ times yields a collection of bootstrap replicates $\{T_{k1}^*, T_{k2}^*, \ldots, T_{kB}^*\}$, which serve as an empirical approximation to the null distribution of $T_k$.

The bootstrap $p$-value for the observed statistic $T_k$ is then computed as
\begin{align}
p_k = \frac{1}{B} \sum_{b=1}^B\indicate(T_k < T_{kb}^*),
\end{align}
where $\indicate(\cdot)$ denotes the indicator function. Intuitively, $p_k$ represents the proportion of bootstrap replicates that exceed the observed statistic, thus quantifying the extremeness of the observed dependence relative to what would be expected under the null hypothesis.

This bootstrap framework has several advantages. First, it does not require explicit knowledge of the limiting distribution of $T_k$, which may be complicated or analytically intractable. Second, it adapts automatically to different distributional settings by relying only on the permutation mechanism. Finally, as the following theorem states, as the number of bootstrap replications $B$ grows, the resulting $p$-value $p_k$ converges to the true significance level, ensuring asymptotic validity of the test.
\begin{theorem}\label{bootstrap}
    Under the IC model given in \eqref{model}, $\Pr_{H_0}(p_k\le\alpha)\le\alpha$ for any desired type \uppercase\expandafter{\romannumeral1} error rate $\alpha\in[0,1]$.
\end{theorem}

\subsection{Diverging parameter}

Next, we show the asymptotic normality of the L-statistic with diverging parameter $\lceil \gamma p^*\rceil$.

\begin{theorem}\label{asymptotic normality}
	Under the IC model given in \eqref{model}, suppose $Z_{11}\sim N(0,1)$, then for any $\gamma\in(0,1]$, under $H_0$, as $\min(n,p)\to\infty$,
	\begin{align*}
    	(p^*)^{-1/2}(T_{\lceil \gamma p^*\rceil}-\mu_{\gamma,n,p})\cd N(0,\sigma_{\gamma}^2),
	\end{align*} 
    where 
    \begin{align*}
        \mu_{\gamma,n,p}:=&\sum_{1\le i<j\le p}\E\left\{(n\hat{\rho}_{ij}^2-v_{\gamma,n})\indicate(n\hat{\rho}_{ij}^2\ge v_{\gamma,n})+\gamma v_{\gamma,n}\right\}\\
        =&p^*\left[\frac{n}{n-1}\left\{1-F_{\text{Beta}(\frac{3}{2},\frac{n-2}{2})}\left(\frac{v_{\gamma,n}}{n}\right)\right\}-v_{\gamma,n}\left\{1-F_{\text{Beta}(\frac{1}{2},\frac{n-2}{2})}\left(\frac{v_{\gamma,n}}{n}\right)\right\}+\gamma v_{\gamma,n}\right],\\
        \sigma_{\gamma}^2:=&3\{1-F_{\chi^2_5}(v_{\gamma})\}-2v_{\gamma}\{1-F_{\chi^2_3}(v_{\gamma})\}+v_{\gamma}^2\{1-F_{\chi^2_1}(v_{\gamma})\}-[1-F_{\chi^2_3}(v_{\gamma})-v_{\gamma}\{1-F_{\chi^2_1}(v_{\gamma})\}]^2<\infty.
\end{align*}
Here, $v_{\gamma,n}$ and $v_{\gamma}$ are the $(1-\gamma)$-th quantiles of $n\text{Beta}(\frac{1}{2},\frac{n-2}{2})$ and $\chi^2_1$, respectively, and $F_{\xi}(\cdot)$ is the C.D.F. of random variable $\xi$.
\end{theorem}
\begin{remark}
    (\romannumeral1) $(p^*)^{-1}\mu_{\gamma,n,p}\to 1-F_{\chi^2_3}(v_{\gamma})-v_{\gamma}\{1-F_{\chi^2_1}(v_{\gamma})\}+\gamma v_{\gamma}<\infty$ as $n\to\infty$. \\
    (\romannumeral2) For $\gamma=1$, $(p^*)^{-1}\mu_{1,n,p}\to 1$ and $\sigma^2_1=2$. The asymptotic normality of $T_{p^*}$ given in Theorem \ref{asymptotic normality} aligns with \cite{schott2005testing}'s result, while relaxing the constraints on $n$ and $p$ to merely require $\min(n,p)\to\infty$.  
\end{remark}

Although the asymptotic null distribution of the diverging-parameter statistic
$T_{\lceil \gamma p^*\rceil}$ is available, we still recommend using the bootstrap procedure
in Section~\ref{fixp} to obtain empirical critical values (and hence evaluate empirical sizes).
In finite samples---especially in high-dimensional settings---the asymptotic approximation
may be poorly calibrated, while the bootstrap typically provides more accurate size control.

Moreover, the additional computational burden is negligible. Since we already implement
the bootstrap to compute critical values (or $p$-values) for the fixed-$k$ statistics $\{T_k\}$,
the same bootstrap resamples can be reused to evaluate $T_{\lceil \gamma p^*\rceil}$.
Therefore, incorporating $T_{\lceil \gamma p^*\rceil}$ into the bootstrap-based size calculation
requires little extra running time.

\subsection{Adaptive procedures}

In practice, we seldom know whether the regime is sparse or dense. In order to adapt to various alternative behaviors, we propose an adaptive test combining different $L$-statistics in \eqref{L-statistic} with different $k$. The key message is that $T_k$ with fixed $k$ and $T_{\lceil \gamma p^*\rceil}$ for $\gamma\in(0,1]$ are asymptotically independent if $H_0$ holds.
\begin{theorem}\label{asymptotic independence}
	Under the IC model given \eqref{model}, suppose $Z_{11}\sim N(0,1)$, then for any $\gamma\in(0,1]$, assume $p=p(n)\to\infty$ and $\log p=o(n^{1/2})$ as $n\to\infty$. Then, as $\min(n,p)\to\infty$, 
	
	\noindent(\romannumeral1) for all integer $1\le s\le p^*$, $\gamma\in(0,1]$ and $x,y\in\mathbb{R}$,
	\begin{align*}
	\Pr_{H_0}\left(n\hat{\rho}^2_{(p^*+1-s)}-b_p\le x,\frac{T_{\lceil \gamma p^*\rceil}-\mu_{\gamma,n,p}}{\sqrt{p^*\sigma_{\gamma}^2}}\le y \right)\to \Lambda(x)\sum_{i=0}^{s-1}\frac{\{\log \Lambda^{-1}(x)\}^i}{i!}\Phi(y),
	\end{align*}
	\noindent(\romannumeral2) for all integer $2\le l\le p^*$, $\gamma\in(0,1]$ and $x_1\ge \cdots\ge x_l\in\mathbb{R}$, $y\in\mathbb{R}$,
	\begin{align*}
	&\Pr_{H_0}\left(\bigcap_{s=1}^l\left(n\hat{\rho}^2_{(p^*+1-s)}-b_p \leq x_s\right), \frac{T_{\lceil \gamma p^*\rceil}-\mu_{\gamma,n,p}}{\sqrt{p^*\sigma_{\gamma}^2}}\le y\right)\\
	\to&\Lambda\left(x_l\right) \sum_{\sum_{i=2}^s l_i \leq s-1, s=2, \ldots, l} \prod_{i=2}^s \frac{\{\log \Lambda^{-1}(x_i)-\log \Lambda^{-1}(x_{i-1})\}^{l_i}}{l_i!}\Phi(y),
	\end{align*}
    where $\Phi(\cdot)$ is the C.D.F. of $N(0,1)$.
\end{theorem}

Based on Theorem~\ref{asymptotic independence}, we construct a Cauchy combination test that aggregates information across multiple choices of $k$:
\begin{align}\label{Cauchy combine test}
T_C
&=\frac{1}{K+1}\tan\!\left\{\Big(\tfrac{1}{2}-p_{5}\Big)\pi\right\}
+\frac{1}{K+1}\sum_{i=1}^K \tan\!\left\{\Big(\tfrac{1}{2}-p_{\lceil 2^{-i}p\rceil}\Big)\pi\right\}.
\end{align}
Under $H_0$, Theorem~\ref{asymptotic independence} implies asymptotic independence between $T_5$ and
$\{T_{\lceil 2^{-i}p\rceil}\}_{i=1}^K$, and each $T_{\lceil 2^{-i}p\rceil}$ is asymptotically normal.
Consequently, $T_C$ is asymptotically standard Cauchy, and the combined $p$-value is
\[
p_C = 1 - G(T_C),
\]
where $G(\cdot)$ denotes the C.D.F. of the standard Cauchy distribution. We reject $H_0$ at level $\alpha$
if $p_C \le \alpha$.

\section{Simulation}

We begin by examining the empirical size performance of the proposed $L$-tests. Specifically, we consider the independent component model
$$
\X_k = \bms^{1/2} \Z_k, \quad k=1,\ldots,n,
$$
where $\Z_k=(Z_{k1},\ldots,Z_{kp})^\top$ and $\{Z_{ki};1\le k\le n,1\le i\le p\}$ are i.i.d. according to one of the following three distributions: (1) standard normal $N(0,1)$; (2) uniform distribution $U(-\sqrt{3},\sqrt{3})$ with unit variance; and (3) standardized $t_5$ distribution, i.e., $t_5/\sqrt{5/3}$. Under the null hypothesis of mutual independence, we set the covariance matrix $\bms=\I_p$.

Table \ref{tab1} summarizes the empirical sizes of the $L$-tests as well as the proposed Cauchy combination test $T_C$, based on sample sizes $n=100,200$ and dimensions $p=100,200$. Across all settings and distributions, the observed rejection rates are very close to the nominal significance level, indicating that both the $L$-tests and the combined test $T_C$ are able to control type I error effectively. These results demonstrate that the proposed procedures are robust to different distributional assumptions, including light-tailed, bounded-support, and heavy-tailed cases, which highlights their practical applicability in a wide range of real data scenarios.

\begin{table}[htbp]
    \centering
    \renewcommand{\arraystretch}{1.3}
    \setlength{\tabcolsep}{8pt}
    \caption{Empirical sizes of L-tests and Cauchy combination test $T_C$.} 
    \begin{tabular}{c |c c |c c |c c} 
        \hline \hline
        & \multicolumn{2}{c}{$N(0,1)$} & \multicolumn{2}{c}{$U(-\sqrt{3},\sqrt{3})$}&\multicolumn{2}{c}{$t_3/\sqrt{3}$}\\ \hline
         {$n$} &  {100} &  {200} &  {100} &  {200} &  {100} &  {200} \\
        \hline
        &\multicolumn{6}{c}{$p=100$}\\ \hline
        $T_5$ & 0.060 & 0.046 & 0.040 & 0.055 & 0.046 & 0.042 \\
        $T_{[\frac{p^*}{256}]}$ & 0.048 & 0.046 & 0.040 & 0.049 & 0.049 & 0.049 \\
        $T_{[\frac{p^*}{128}]}$ & 0.055 & 0.046 & 0.039 & 0.055 & 0.053 & 0.054 \\
        $T_{[\frac{p^*}{64}]}$ & 0.048 & 0.051 & 0.043 & 0.058 & 0.052 & 0.054 \\
        $T_{[\frac{p^*}{32}]}$ & 0.050 & 0.049 & 0.049 & 0.053 & 0.051 & 0.057 \\
        $T_{[\frac{p^*}{16}]}$ & 0.045 & 0.048 & 0.060 & 0.067 & 0.057 & 0.053 \\
        $T_{[\frac{p^*}{8}]}$ & 0.045 & 0.054 & 0.059 & 0.064 & 0.058 & 0.053 \\
        $T_{[\frac{p^*}{4}]}$ & 0.043 & 0.055 & 0.064 & 0.064 & 0.058 & 0.056 \\
        $T_{[\frac{p^*}{2}]}$ & 0.045 & 0.055 & 0.065 & 0.062 & 0.056 & 0.050 \\
        $T_C$ & 0.043 & 0.065 & 0.063 & 0.067 & 0.061 & 0.067 \\
        \hline 
        &\multicolumn{6}{c}{$p=200$}\\ \hline
$T_{[5]}$&0.047&0.058&0.059&0.051&0.052&0.048\\
$T_{[\frac{p^*}{624}]}$&0.057&0.066&0.056&0.045&0.051&0.063\\
$T_{[\frac{p^*}{312}]}$&0.055&0.064&0.059&0.04&0.054&0.058\\
$T_{[\frac{p^*}{256}]}$&0.044&0.064&0.052&0.043&0.053&0.052\\
$T_{[\frac{p^*}{128}]}$&0.046&0.063&0.045&0.036&0.056&0.056\\
$T_{[\frac{p^*}{64}]}$&0.051&0.068&0.042&0.039&0.059&0.052\\
$T_{[\frac{p^*}{32}]}$&0.056&0.067&0.036&0.044&0.067&0.059\\
$T_{[\frac{p^*}{16}]}$&0.049&0.066&0.039&0.042&0.057&0.053\\
$T_{[\frac{p^*}{8}]}$&0.043&0.067&0.044&0.043&0.061&0.052\\
$T_{[\frac{p^*}{4}]}$&0.045&0.066&0.048&0.042&0.065&0.047\\
$T_{[\frac{p^*}{2}]}$&0.043&0.061&0.039&0.039&0.062&0.043\\
$T_C$&0.064&0.057&0.058&0.058&0.066&0.062\\ \hline \hline
    \end{tabular}
    \label{tab1} 
\end{table}

Next, we compare the proposed adaptive method $T_C$ with several existing procedures: 
the $L_2$-norm based sum-type test $T_{SC}$ of \cite{schott2005testing}, 
the $L_4$-norm based sum-type test $T_{LX}$ of \cite{li2023test}, 
the max-type test $T_J$ of \cite{jiang2004asymptotic}, 
and the adaptive test $T_F$ of \cite{feng2022max}. 
To ensure a fair comparison, we adopt a size-corrected power analysis. Specifically, 
we first compute the critical values of each test under the null hypothesis, 
and then evaluate their empirical power under the alternative hypothesis.  

For the alternative, we consider the block covariance structure
\[
\bms=\begin{pmatrix}
\bms_1 & \bm 0 \\
\bm 0 & \I_{p-m}
\end{pmatrix},
\]
where $\bms_1=(\sigma_{ij})_{1\leq i,j\leq m}\in\mathbb{R}^{m\times m}$ with entries
\[
\sigma_{ij}=\left\{\frac{\theta}{m(m-1)}\right\}^{|i-j|/4}.
\]
Here we set $\theta=1.5$. 
Under this setting, the sparsity level of the alternative is given by $s=m(m-1)/2$.  

Figure \ref{fig1}-\ref{fig2} report the empirical power curves of the different methods across a range of sparsity levels $s$. 
Several key observations emerge. First, our proposed adaptive test $T_C$ consistently achieves the best or near-best performance across most regimes. 
In particular, when the alternative is moderately sparse, $T_C$ significantly outperforms all competitors. 
Second, in the very dense regime ($s$ large), the $L_2$-based sum-type test $T_{SC}$ attains slightly higher power than $T_C$, 
as expected due to its design focus on dense alternatives. 
Third, the max-type test $T_J$ is only effective when the alternative is extremely sparse ($s$ very small), 
but its power deteriorates rapidly as sparsity decreases. 
Fourth, the adaptive procedure $T_F$, which combines $T_{SC}$ and $T_J$, improves robustness across sparsity levels; 
however, it still fails to fully capture intermediate alternatives and is uniformly less powerful than $T_C$. 
Finally, the $L_4$-norm based test $T_{LX}$ performs comparably to $T_F$ in light- or moderate-tailed distributions, 
but under heavy-tailed settings, e.g.\ $Z_{ki}\sim t(5)/\sqrt{5/3}$, 
its performance drops sharply due to the requirement of finite eighth moments.  

Overall, these results demonstrate that the proposed $L$-test with Cauchy combination, $T_C$, 
provides a highly efficient and robust testing procedure. 
Unlike existing methods that excel only in specific sparsity regimes or distributional settings, 
$T_C$ maintains stable and superior power across a wide spectrum of alternatives, 
making it particularly suitable for high-dimensional applications with unknown dependence structures and tail behaviors.

\begin{figure}[htpb]
	\centering
 \caption{Power curves of each test with $n=100,p=100$.}
	\includegraphics[width=\linewidth]{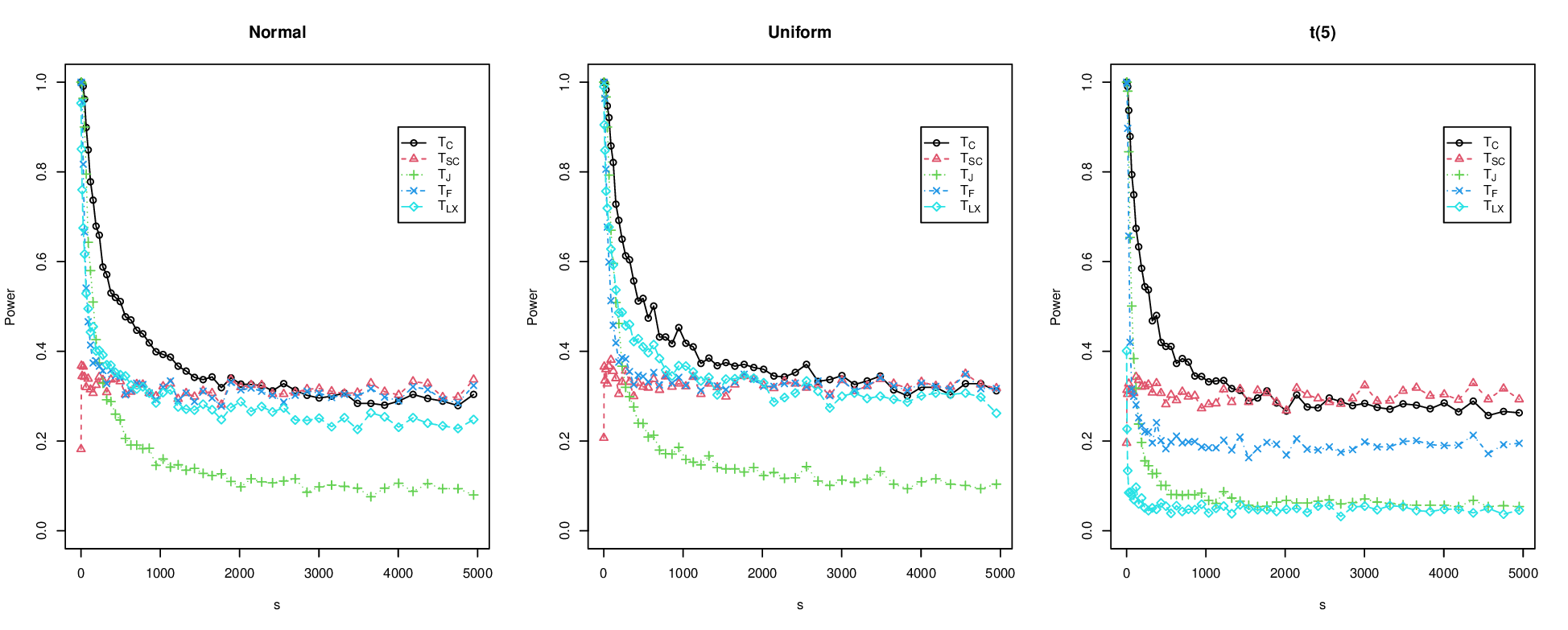} \label{fig1}
\end{figure}

\begin{figure}[htpb]
	\centering
	\includegraphics[width=\linewidth]{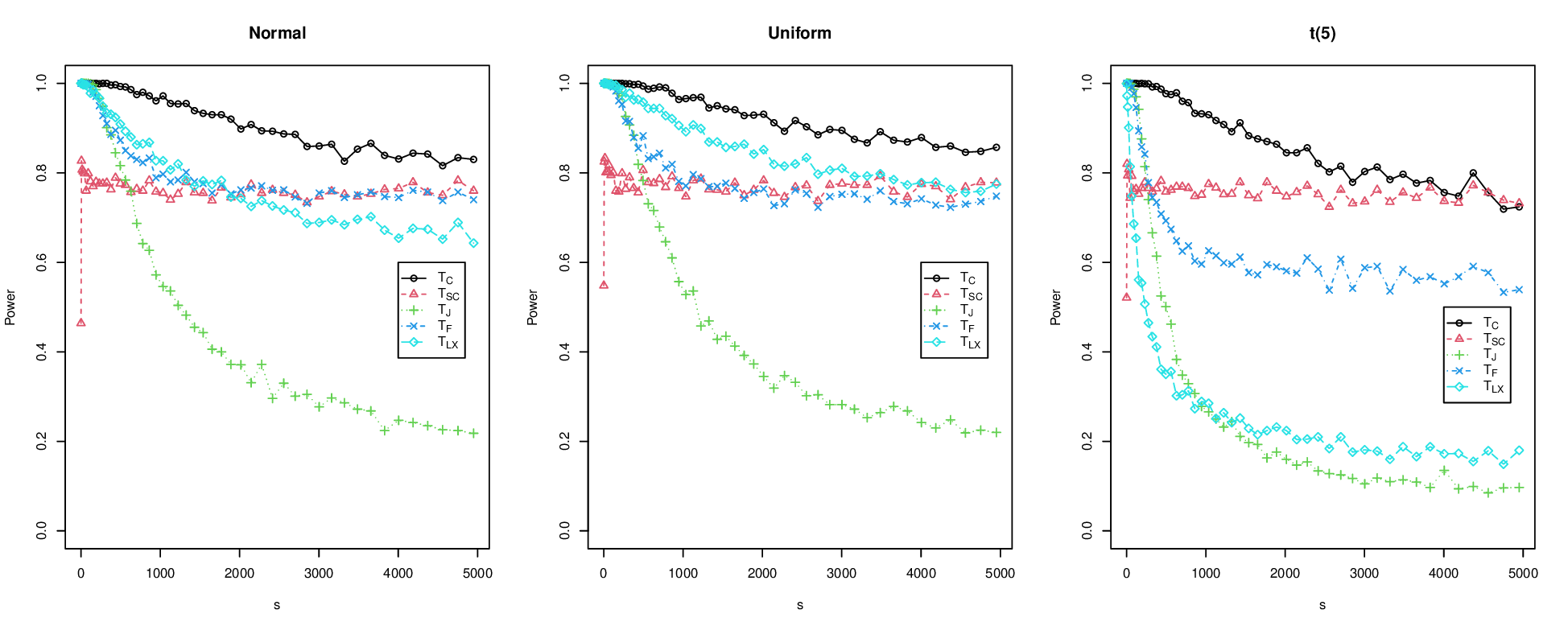}
	\caption{Power curves of each test with $n=200,p=100$.}\label{fig2}
\end{figure}

\section{Conclusion}
In this paper, we proposed two types of $L$-tests for high-dimensional mutual independence testing and further introduced a Cauchy combination procedure that integrates different $L$-tests. We established their theoretical properties and demonstrated through simulation studies that our method outperforms existing approaches based on Pearson’s correlation. Although the proposed tests exhibit strong performance, they still rely on Pearson’s correlation as the baseline measure, which is known to be sensitive to heavy-tailed distributions and incapable of capturing nonlinear dependencies. Recently, a growing body of work has focused on high-dimensional independence testing based on nonlinear correlation measures, such as \cite{han2017distribution}, \cite{yao2018testing}, \cite{drton2020high}, \cite{cai2024asymptotic}, \cite{wang2024rank}, and \cite{xia2025consistent}. An important avenue for future research is to explore how $L$-type test statistics can be constructed upon these nonlinear correlation measures in order to achieve robustness against heavy tails and nonlinear dependence structures.

\section{Appendix}
In this section we prove Theorems \ref{asymptotic distribution of ordered statistics}-\ref{asymptotic independence}. The letter $C$ stands for a constant that may vary from place to place throughout this section.

We begin by collecting a few essential technical lemmas in Section \ref{appendix1}.

\subsection{Technical tools}\label{appendix1}
Intuitively, by a point process, we generally mean a series of events occurring in time (or space, or both) according to some statistical law. Suppose that $\{N_n\}$ is a sequence of point processes on a rectangle $S\subset\mathbb{R}^n$ and that $N$ is a point process. Then we may say that $N_n$ converges in distribution to $N$ (written $N_n\cd N$) if the sequence of vector r.v.'s $(N_n(B_1),\ldots,N_n(B_k))$ converges in distribution to $(N(B_1),\ldots,N(B_k))$ for each choice of $k$, and all bounded Borel sets $B_i\subset S$ such that $N(\partial B_i) = 0$ a.s., $i = 1,\ldots,k$, (writing $\partial B$ for the boundary of the set $B$).
\begin{lemma}\label{lemma1}
    (\romannumeral1) Let $N_n, n\ge 1$ and $N$ be point processes on the semiclosed interval $S$ in the real line, $N$ being simple. Suppose that\\
    (a) $\E(N_n((c,d]))\to\E(N((c,d]))$ for all $-\infty<c<d<\infty$ such that $(c,d]\subset S$, and\\
    (b) $\Pr(N_n(B)=0)\to\Pr(N(B)=0)$ for all $B$ of the form $\cup_{i=1}^k(c_i,d_i]$ with $(c_i,d_i]\subset S$ for $i=1,\ldots,k;k\ge 1$.\\
    Then $N_n\cd N$.\\
    (\romannumeral2) The same is true for point processes on a semi-closed rectangle $S$ in the plane, if the semi-closed intervals $(c,d],(c_i,d_i]$ are replaced by semiclosed rectangles $(c,d]\times(\gamma,\delta],(c_i,d_i]\times(\gamma_i,\delta_i]$.
\end{lemma}
\begin{proof}
    This basic Point Process Theory follows from Theorem A.1 in \cite{Leadbetter1983ExtremesAR}.
\end{proof}

\begin{lemma}\label{lemma2}
    Suppose $\{Z_{ki};1\le k\le n,1\le i\le p\}$ are i.i.d. random variables with $E(Z_{11})=0$, $\E(Z_{11}^2)=1$ and $\E\{\exp(t_0|Z_{11}|^{a})\}<\infty$ for some $0<a\le 2$ and $t_0>0$. Assume $p=p_n\to \infty$ and $\log p=o(n^\beta)$ as $n\to\infty$, where $\beta=a/(4+a)$. Define
    \begin{align*}
        r_{ij}=\frac{(\Z_{.i}-\bar{Z}_i)^\top(\Z_{.j}-\bar{Z}_j)}{\Vert\Z_{.i}-\bar{Z}_i\Vert_2\cdot\Vert\Z_{.j}-\bar{Z}_j\Vert_2},
    \end{align*}
    where $\Z_{.i}=(Z_{1i},\ldots,Z_{ni})^\top$ and $\bar{Z}_i=n^{-1}\sum_{k=1}^nZ_{ki}$. Here, we write $\Z_{.i}-\bar{Z}_i$ for $\Z_{.i}-\bar{Z}_i\bone_n$, where $\bone_n=(1,\ldots,1)^\top\in\mR^n$. Then $\max_{1\le i<j\le p}nr_{ij}^2-4\log p+\log(\log p)$ converges weakly to an extreme distribution of type \uppercase\expandafter{\romannumeral1} with distribution function $\Lambda(x)=\exp\{-(8\pi)^{-1/2}\exp (-x/2)\}$.
\end{lemma}
\begin{proof}
    See Theorem 3 in \cite{cai2011limiting}.
\end{proof}

\begin{lemma}\label{lemma3}
    Let $I(1),\ldots,I(d)$ be disjoint index sets and $\{\eta_{\alpha},\alpha\in\cup_{j=1}^dI(j)\}$ be random variables. For each $1\le j\le d$, let $\{B_{\alpha},\alpha\in I(j)\}$ be a set of subsets of $I(j)$, also set $\lambda_j=\sum_{\alpha\in I(j)}\Pr(\eta_{\alpha}>t_j)$ for some $t_j\in\mathbb{R}$. Then
    \begin{align*}
        \left|\Pr\left(\bigcap_{j=1}^d\left\{\max_{\alpha\in I(j)}\eta_{\alpha}\le t_j\right\}\right)-\exp\left(-\sum_{j=1}^d\lambda_j\right)\right|\le 2\left\{1\wedge 1.4\left(\min_{1\le j\le d}\lambda_j\right)^{-1/2}\right\}(2b_1+2b_2+b_3),
    \end{align*}
    where 
    \begin{align*}
        b_1=&\sum_{j=1}^d\sum_{\alpha\in I(j)}\sum_{\beta\in B_\alpha}\Pr(\eta_\alpha>t_j)\Pr(\eta_\beta>t_j),~~b_2=\sum_{j=1}^d\sum_{\alpha\in I(j)}\sum_{\alpha\ne\beta\in B_\alpha}\Pr(\eta_\alpha>t_j,\eta_\beta>t_j),\\
        b_3=&\sum_{j=1}^d\sum_{\alpha\in I(j)}\E\{|\Pr(\eta_\alpha>t_j|\sigma\{\eta_\beta:\beta\in I(j)-B_\alpha\})-\Pr(\eta_\alpha>t_j)|\},
    \end{align*}
    and $\sigma\{\eta_\beta:\beta\in I(j)-B_\alpha\}$ is the $\sigma$-algebra generated by $\{\eta_\beta:\beta\in I(j)-B_\alpha\}$. In particular, if for each $1\le j\le d$ and all $\alpha\in I(j)$, $\eta_\alpha$ is independent of $\{\eta_\beta:\beta\in I(j)-B_\alpha\}$, then $b_3=0$.
\end{lemma}
\begin{proof}
    This Poisson approximation result is essentially a special case of Theorem 2 from \cite{AGG89poisson_approximation}.
\end{proof}

\begin{lemma}\label{lemma4}
    Let $\xi_1,\ldots,\xi_n$ be i.i.d. random variables with $\E(\xi_1)=0$, $\E(\xi_1^2)=1$ and $\E(e^{t_0|\xi_1|^a})<\infty$ for some $t_0>0$ and $0< a\le 2$. Put $S_n=\sum_{i=1}^n\xi_i$ and $\beta=a/(2+a)$. Then, for any $\{p_n:n\ge 1\}$ with $0<p_n\to\infty$ and $\log p_n=o(n^\beta)$ and $\{y_n:n\ge 1\}$ with $y_n\to y>0$,
    \begin{align*}
        \Pr\left(\frac{S_n}{\sqrt{n\log p_n}}\ge y_n\right)\sim\frac{p_n^{-y_n^2/2}(\log p_n)^{-1/2}}{\sqrt{2\pi}y}
    \end{align*}
    as $n\to\infty$.
\end{lemma}
\begin{proof}
    See Lemma 6.8 in \cite{cai2011limiting}.
\end{proof}

\begin{lemma}\label{lemma5}
    Suppose $\{Z_{ki};1\le k\le n,1\le i\le p\}$ are i.i.d. random variables with $E(Z_{11})=0$, $\E(Z_{11}^2)=1$ and $\E\{\exp(t_0|Z_{11}|^{a})\}<\infty$ for some $0<a\le 2$ and $t_0>0$. Assume $p=p_n\to \infty$ and $\log p=o(n^\beta)$ as $n\to\infty$, where $\beta=a/(4+a)$. Then, for any $\varepsilon>0$ and a sequence of positive numbers $\{t_n\}$ with limit $t>0$,
    \begin{align*}
        \Psi_n:=\E\left\{\Pr^1\left(\left|\sum_{k=1}^nZ_{k1}Z_{k2}\right|>t_n\sqrt{n\log p}\right)^2\right\}=O\left(\frac{1}{p^{t^2-\varepsilon}}\right)
    \end{align*}
    as $n\to\infty$, where $\Pr^1$ stands for the conditional probability given $\{Z_{k1}:1\le k\le n\}$.
\end{lemma}
\begin{proof}
    See Lemma 6.7 in \cite{cai2011limiting}.
\end{proof}

The following is known: if $\{X_n\}$ are tight, then for any sequence of constants $\{\varepsilon_n\}$ with $\lim_{n\to\infty}\varepsilon_n=0$, we have $\varepsilon_nX_n\to 0$ in probability as $n\to\infty$. 
\begin{lemma}\label{lemma6}
    Suppose $\{Z_{ki};1\le k\le n,1\le i\le p\}$ are i.i.d. random variables with $E(Z_{11})=0$ and $\E(Z_{11}^2)=1$. Set $\Z_{.i}=(Z_{1i},\ldots,Z_{ni})^\top$, $\bar{Z}_i=n^{-1}\sum_{k=1}^nZ_{ki}$ and $h_i=\Vert\Z_{.i}-\bar{Z}_i\Vert_2/\sqrt{n}$ for each $i$. Define 
    \begin{align*}
        b_{n,1}=\max_{1\le i\le p}|h_i-1|,~~b_{n,3}=\min_{1\le i\le p}h_i~~\text{and}~~b_{n,4}=\max_{1\le i\le p}|\bar{Z}_i|.
    \end{align*}
    Then, $b_{n,3}\to 1$ in probability as $n\to\infty$, and $\{\sqrt{n/\log p}b_{n,1}\}$ and $\{\sqrt{n/\log p}b_{n,4}\}$ are tight provided one of the following conditions holds:\\
    (\romannumeral1) $|Z_{11}|\le C$ for some constant $C>0$, $p_n\to\infty$ and $\log p_n=o(n)$ as $n\to\infty$;\\
    (\romannumeral2) $\E\{\exp(t_0|Z_{11}|^{a})\}<\infty$ for some $a>0$ and $t_0>0$, and $p_n\to \infty$ with $\log p=o(n^\beta)$ as $n\to\infty$, where $\beta=a/(4+a)$.
\end{lemma}
\begin{proof}
    See Lemma 6.5 in \cite{cai2011limiting}.
\end{proof}

\begin{lemma}\label{lemma7}
    Let $\P=\I_n-n^{-1}\bone_n\bone_n^\top$. Let $\O$ be a $n\times n$ orthogonal matrix such that
    \begin{align*}
        \P=\O\left(
        \begin{array}{cc}
            \I_{n-1} & \bzero \\
            \bzero & \bzero 
        \end{array}
        \right)\O^\top.
    \end{align*} 
    Write $\O=(\U,\V)$, where $\U$ is a $n\times(n-1)$ submatrix. Let $\{Z_{ki};1\le k\le n,1\le i\le p\}$ be i.i.d. random variables with $Z_{11}\sim N(0,1)$. Write $\Z_{.i}=(Z_{1i},\ldots,Z_{ni})^\top$ for each $i$. Let $\s_1,\ldots,\s_p$ be i.i.d. random vectors uniformly distributed on $\mS^{n-2}$. Then $\left(\frac{\P\Z_{.1}}{\Vert\P\Z_{.1}\Vert_2},\ldots,\frac{\P\Z_{.p}}{\Vert\P\Z_{.p}\Vert_2}\right)$ and $(\U\s_1,\ldots,\U\s_p)$ have the same distribution.
\end{lemma}
\begin{proof}
    By the orthogonality, we have
    \begin{align*}
        \U^{\top}\U=\I_{n-1}~~\text{and}~~\U\U^{\top}=\P.
    \end{align*}
    By the orthogonal invariance of normal distributions, we have
    \begin{align*}
        \P\Z_{.i}=(\U,\bzero)\O^\top\Z_{.i}\stackrel{d}{=}(\U,\bzero)\Z_{.i}=\U\Z_{-n,i},
    \end{align*}
    where $\Z_{-n,i}=(Z_{1i},\ldots,Z_{n-1,i})^\top$. Then $\frac{\P\Z_{.i}}{\Vert\P\Z_{.i}\Vert_2}$ as a function of $\P\Z_{.i}$, has the same distribution as that of
    \begin{align*}
        \frac{\U\Z_{-n,i}}{\Vert\U\Z_{-n,i}\Vert_2}=\U\frac{\Z_{-n,i}}{\Vert\Z_{-n,i}\Vert_2},
    \end{align*}
    for each $i$. The desired conclusion then follows from the independence among $\Z_{.1},\ldots,\Z_{.p}$.
\end{proof}

\begin{lemma}\label{lemma8}
    Let $\s_1,\ldots,\s_p$ be i.i.d. random vectors uniformly distributed on $\mS^{n-2}$. Let
    \begin{align*}
        F_n(x):=\Pr(n(\s_1^\top\s_2)^2\le x),~~\hat{F}_{n,p}(x):=\frac{2}{p(p-1)}\sum_{1\le i<j\le p}\indicate(n(\s_i^\top\s_j)^2\le x),
    \end{align*}
    and for some $\gamma\in(0,1]$,
    \begin{align*}
        v_{\gamma,n}:=\inf\{x:F_n(x)\ge 1-\gamma\},~~\hat{v}_{\gamma,n,p}:=\inf\{x:\hat{F}_{n,p}(x)\ge 1-\gamma\}.
    \end{align*}
    Then, for each $n$, $\hat{v}_{\gamma,n,p}-v_{\gamma,n}=O_p(p^{-1})$.
\end{lemma}
\begin{proof}
    By Theorem 1.5.7(i) and the argument for (5) on p.147 of \cite{Muirhead1982AspectsOM},
    \begin{align}\label{eql1}
        (\s_i^\top\s_j)^2\sim\text{Beta}\left(\frac{1}{2},\frac{n-2}{2}\right),~~\text{for any}~~1\le i<j\le p.
    \end{align}
    Then, the density of $n(\s_1^\top\s_2)^2$ is given by 
    \begin{align*}
        f_n(x)=\frac{\Gamma(\frac{n-1}{2})}{\Gamma(\frac{1}{2})\Gamma(\frac{n-2}{2})}n^{-1/2}x^{-1/2}(1-x^2)^{(n-4)/2},~~0\le x\le n.
    \end{align*}
    For any $\varepsilon>0$, let $C_1=\inf_{x\in[v_{\gamma,n}-\varepsilon,v_{\gamma,n}+\varepsilon]}f_n(x)$. By definitions,
    \begin{align}\label{eql2}
        &\Pr(|\hat{v}_{\gamma,n,p}-v_{\gamma,n}|\ge\varepsilon)\nonumber\\
        =&\Pr(\hat{v}_{\gamma,n,p}\ge v_{\gamma,n}+\varepsilon)+\Pr(\hat{v}_{\gamma,n,p}\le v_{\gamma,n}-\varepsilon)\nonumber\\
        \le&\Pr(\hat{F}_{n,p}(v_{\gamma,n}+\varepsilon)\le 1-\gamma)+\Pr(\hat{F}_{n,p}(v_{\gamma,n}-\varepsilon)\ge 1-\gamma)\nonumber\\
        \le&\Pr(|\hat{F}_{n,p}(v_{\gamma,n}+\varepsilon)-F_{n}(v_{\gamma,n}+\varepsilon)|>C_1\varepsilon)+\Pr(|\hat{F}_{n,p}(v_{\gamma,n}-\varepsilon)-F_{n}(v_{\gamma,n}-\varepsilon)|>C_1\varepsilon).
    \end{align}
    For each $x\in\mR$, define $J_j:=\sum_{i=1}^{j-1}\indicate(n(\s_i^\top\s_j)^2\le x)$. Then
    \begin{align*}
        \sqrt{p^*}\{\hat{F}_{n,p}(x)-F_n(x)\}=\frac{1}{\sqrt{p^*}}\sum_{j=2}^p\{J_j-\E(J_j)\}.
    \end{align*}
    We first claim that
    \begin{align}\label{eql3}
        \{\s_1^\top\s_2,\s_1^\top\s_3,\ldots,\s_1^\top\s_p\}~~\text{are independent}.
    \end{align}
    In fact, let $\s$ be uniformly distributed on $\mS^{n-2}$ . Then, $\x^\top\s$ has the same distribution as that of $(1,0,\ldots,0)^\top\s$ for any $\x\in\mS^{n-2}$ (see, e.g., Theorem 1.5.7(i) and the argument for (5) on p.147 of \cite{Muirhead1982AspectsOM}). Since $\s_1,\ldots,\s_p$ are i.i.d. random vectors, we know that, conditioning on $\s_1$, the random variables $\{\s_1^\top\s_2,\s_1^\top\s_3,\ldots,\s_1^\top\s_p\}$ are i.i.d. with a common distribution of $(1,0,\ldots,0)^\top\s$. In particular, their conditional distributions do not depend on $\s_1$. This proves the claim. Consequently, it is easy to calculate that
    \begin{align*}
        &\E(J_j)=(j-1)F_{\text{Beta}\left(\frac{1}{2},\frac{n-2}{2}\right)}\left(\frac{x}{n}\right),1\le j\le p\\
        &\E(J_j^2)=(j-1)F_{\text{Beta}\left(\frac{1}{2},\frac{n-2}{2}\right)}\left(\frac{x}{n}\right)+(j-1)(j-2)\left\{F_{\text{Beta}\left(\frac{1}{2},\frac{n-2}{2}\right)}\left(\frac{x}{n}\right)\right\}^2,1\le j\le p,~~\text{and}\\
        &\E(J_{j_1}J_{j_2})=(j_1-1)(j_2-1)\left\{F_{\text{Beta}\left(\frac{1}{2},\frac{n-2}{2}\right)}\left(\frac{x}{n}\right)\right\}^2,1\le j_1\ne j_2\le p.
    \end{align*}
    Thus,
    \begin{align*}
        \E\left[\frac{1}{\sqrt{p^*}}\sum_{j=2}^p\{J_j-\E(J_j)\}\right]=0,
    \end{align*}
    and 
    \begin{align*}
        \E\left[\left|\frac{1}{\sqrt{p^*}}\sum_{j=2}^p\{J_j-\E(J_j)\}\right|^2\right]=&\frac{1}{p^*}\sum_{j=2}^{p}\var(J_j)+\frac{1}{p^*}\sum_{2\le j_1\ne j_2\le p}\cov(J_{j_1},J_{j_2})\\
        =&F_{\text{Beta}\left(\frac{1}{2},\frac{n-2}{2}\right)}\left(\frac{x}{n}\right)-\left\{F_{\text{Beta}\left(\frac{1}{2},\frac{n-2}{2}\right)}\left(\frac{x}{n}\right)\right\}^2.
    \end{align*}
    By Gamma approximation for Beta tails, we have
    \begin{align*}
        F_{\text{Beta}\left(\frac{1}{2},\frac{n-2}{2}\right)}\left(\frac{x}{n}\right)\to F_{\text{Gamma}\left(\frac{1}{2},1\right)}(x/2)=F_{\chi^2_1}(x),~~\text{as}~~n\to\infty,
    \end{align*}
    which implies for each $n$, 
    \begin{align*}
        \E\left[\left|\frac{1}{\sqrt{p^*}}\sum_{j=2}^p\{J_j-\E(J_j)\}\right|^2\right]=O(1).
    \end{align*}
    Hence,
    \begin{align*}
        \sqrt{p^*}\{\hat{F}_{n,p}(x)-F_n(x)\}=\frac{1}{\sqrt{p^*}}\sum_{j=2}^p\{J_j-\E(J_j)\}=O_p(1),
    \end{align*}
    which, combining with \eqref{eql2}, leads to the desired result.    
\end{proof}

\begin{lemma}\label{lemma9}
    Let $\s_1,\s_2$ be i.i.d. random vectors uniformly distributed on $\mS^{n-2}$. For any $x\in\mR^+$, define 
    \begin{align*}
        A_{l,n}(x):=\E[\{n(\s_1^\top\s_2)^2-x\}^l\indicate\{n(\s_i^\top\s_j)^2\ge x\}],l=1,2,3,4.
    \end{align*}
    Then, 
    \begin{align*}
        A_{1,n}(x)=&\frac{n}{n-1}\left\{1-F_{\text{Beta}\left(\frac{3}{2},\frac{n-2}{2}\right)}\left(\frac{x}{n}\right)\right\}-x\left\{1-F_{\text{Beta}\left(\frac{1}{2},\frac{n-2}{2}\right)}\left(\frac{x}{n}\right)\right\},\\
        A_{2,n}(x)=&\frac{3n^2}{(n-1)(n+1)}\left\{1-F_{\text{Beta}\left(\frac{5}{2},\frac{n-2}{2}\right)}\left(\frac{x}{n}\right)\right\}-\frac{2xn}{n-1}\left\{1-F_{\text{Beta}\left(\frac{3}{2},\frac{n-2}{2}\right)}\left(\frac{x}{n}\right)\right\}\\
        &+x^2\left\{1-F_{\text{Beta}\left(\frac{1}{2},\frac{n-2}{2}\right)}\left(\frac{x}{n}\right)\right\},\\
        A_{3,n}(x)=&\frac{15n^3}{(n-1)(n+1)(n+3)}\left\{1-F_{\text{Beta}\left(\frac{7}{2},\frac{n-2}{2}\right)}\left(\frac{x}{n}\right)\right\}-\frac{9xn^2}{(n-1)(n+1)}\left\{1-F_{\text{Beta}\left(\frac{5}{2},\frac{n-2}{2}\right)}\left(\frac{x}{n}\right)\right\}\\
        &+\frac{3x^2n}{n-1}\left\{1-F_{\text{Beta}\left(\frac{3}{2},\frac{n-2}{2}\right)}\left(\frac{x}{n}\right)\right\}-x^3\left\{1-F_{\text{Beta}\left(\frac{1}{2},\frac{n-2}{2}\right)}\left(\frac{x}{n}\right)\right\},\\
        A_{4,n}(x)=&\frac{105n^4}{(n-1)(n+1)(n+3)(n+5)}\left\{1-F_{\text{Beta}\left(\frac{9}{2},\frac{n-2}{2}\right)}\left(\frac{x}{n}\right)\right\}\\
        &-\frac{60xn^3}{(n-1)(n+1)(n+3)}\left\{1-F_{\text{Beta}\left(\frac{7}{2},\frac{n-2}{2}\right)}\left(\frac{x}{n}\right)\right\}+\frac{18x^2n^2}{(n-1)(n+1)}\left\{1-F_{\text{Beta}\left(\frac{5}{2},\frac{n-2}{2}\right)}\left(\frac{x}{n}\right)\right\}\\
        &-\frac{4x^3n}{n-1}\left\{1-F_{\text{Beta}\left(\frac{3}{2},\frac{n-2}{2}\right)}\left(\frac{x}{n}\right)\right\}+x^4\left\{1-F_{\text{Beta}\left(\frac{1}{2},\frac{n-2}{2}\right)}\left(\frac{x}{n}\right)\right\}.
    \end{align*}
    In addition, $A_{l,n}(x)\to A_l(x)$ for $l=1,2,3,4$, as $n\to\infty$, where 
    \begin{align*}
        A_{1}(x):=&\{1-F_{\chi^2_3}(x)\}-x\{1-F_{\chi^2_1}(x)\},\\
        A_{2}(x):=&3\{1-F_{\chi^2_5}(x)\}-2x\{1-F_{\chi^2_3}(x)\}+x^2\{1-F_{\chi^2_1}(x)\},\\
        A_{3}(x):=&15\{1-F_{\chi^2_7}(x)\}-9x\{1-F_{\chi^2_5}(x)\}+3x^2\{1-F_{\chi^2_3}(x)\}-x^3\{1-F_{\chi^2_1}(x)\},\\
        A_{4}(x):=&105\{1-F_{\chi^2_9}(x)\}-60x\{1-F_{\chi^2_7}(x)\}+18x^2\{1-F_{\chi^2_5}(x)\}-4x^3\{1-F_{\chi^2_3}(x)\}+x^4\{1-F_{\chi^2_1}(x)\}.
    \end{align*}
\end{lemma}
\begin{proof}
    Using the fact that \eqref{eql1} and \eqref{eql3} and simple computation, we can calculate $A_{l,n}(x),l=1,2,3,4$. By Gamma approximation for Beta tails, we obtain the limiting results.
\end{proof}

\begin{lemma}\label{lemma10}
    Let $\s_1,\ldots,\s_p$ be i.i.d. random vectors uniformly distributed on $\mS^{n-2}$. Set $\mathcal{F}_0=\{\emptyset,\Omega\}$ and $\mathcal{F}_j=\sigma(\s_1,\ldots,\s_j)$, where $\sigma(\s_1,\ldots,\s_j)$ is the $\sigma$-algebra generated by $\{\s_1,\ldots,\s_j\}$ for $1\le j\le p$. Here $\Omega$ is the sample space on which random variables $\{\s_i;1\le i\le p\}$ are defined on.
    For any $x\in\mR^+$, define $I_{x,1}=0$ and for $2\le j\le p$,
    \begin{align*}
        I_{x,j}:=\sum_{i=1}^{j-1}\{n(\s_i^\top\s_j)^2-x\}\indicate\{n(\s_i^\top\s_j)^2\ge x\},~~\text{and}~~Y_{x,j}:=I_{x,j}-\E(I_{x,j}|\mathcal{F}_{j-1}).
    \end{align*}
    Then, as $\min(n,p)\to\infty$,
    \begin{align*}
        \frac{1}{\sqrt{p^*}}\sum_{j=2}^pY_{x,j}\cd N(0,\sigma_{x,\gamma}^2),
    \end{align*}
    where $\E(I_{x,j}|\mathcal{F}_{j-1})=(j-1)A_{1,n}(x)$ and $\sigma_{x,\gamma}^2:=A_2(x)-\{A_1(x)\}^2$, with $A_{l,n}(x)$ and $A_l(x)$ for $l=1,2$ being defined in Lemma \ref{lemma9}.
\end{lemma}
\begin{proof}
    Note that $\{Y_{x,j};2\le j\le p\}$ forms a sequence of martingale difference with respect to $\{\mathcal{F}_j;2\le j\le p\}$. Then, in order to show the asymptotic normality, we will employ the Lindeberg-Feller central limit theorem (see, for example, p. 344 from \cite{durrett2019probability}). To achieve so, it is enough to verify that
    \begin{align}\label{eql4}
        \frac{1}{p^*}\sum_{j=2}^p\E(Y_{x,j}^2|\mathcal{F}_{j-1})\cp\sigma^2_{x,\gamma},
    \end{align}
    and
    \begin{align}\label{eql5}
        \frac{1}{(p^*)^2}\sum_{j=2}^p\E(Y_{x,j}^4|\mathcal{F}_{j-1})\cp0.
    \end{align}

    \noindent\textit{Step 1.} We prove \eqref{eql4}. By \eqref{eql1} and \eqref{eql3} and Lemma \ref{lemma9}, recalling the definitions of $A_{l,n}(x)$ and $A_l(x)$ given in Lemma \ref{lemma9}, it is easy to calculate 
    \begin{align*}
        \E(I_{x,j}|\mathcal{F}_{j-1})=&\sum_{i=1}^{j-1}\E[\{n(\s_i^\top\s_j)^2-x\}\indicate\{n(\s_i^\top\s_j)^2\ge x\}|\s_i]=(j-1)A_{1,n}(x)=\E(I_{x,j}),
    \end{align*}
    \begin{align}\label{eql6}
        &\E(I_{x,j}^2|\mathcal{F}_{j-1})\nonumber\\
        =&\sum_{i=1}^{j-1}\E[\{n(\s_i^\top\s_j)^2-x\}^2\indicate\{n(\s_i^\top\s_j)^2\ge x\}|\s_i]\nonumber\\
        &+2\sum_{1\le i_1<i_2\le j-1}\E[\{n(\s_{i_1}^\top\s_j)^2-x\}\indicate\{n(\s_{i_1}^\top\s_j)^2\ge x\}\{n(\s_{i_2}^\top\s_j)^2-x\}\indicate\{n(\s_{i_2}^\top\s_j)^2\ge x\}|\s_{i_1},\s_{i_2}]\nonumber\\
        =&(j-1)A_{2,n}(x)+2\sum_{1\le i_1<i_2\le j-1}\E\{h_x(\sqrt{n}\s_{i_1}^\top\s_j)h_x(\sqrt{n}\s_{i_2}^\top\s_j)|\s_{i_1},\s_{i_2}\},
    \end{align}
    where $h_x(u):=(u^2-x)\indicate(u^2\ge x)$, and
    \begin{align*}
        \E(I_{x,j}^2)=(j-1)A_{2,n}(x)+(j-1)(j-2)\{A_{1,n}(x)\}^2.
    \end{align*}
    Hence,
    \begin{align}\label{eql7}
        \E\left\{\frac{1}{p^*}\sum_{j=2}^p\E(Y_{x,j}^2|\mathcal{F}_{j-1})\right\}=&\frac{1}{p^*}\sum_{j=2}^p\var(I_{x,j})=A_{2,n}(x)-\{A_{1,n}(x)\}^2\to A_2(x)-\{A_1(x)\}^2=:\sigma_{x,\gamma}^2.
    \end{align}
    Furthermore, 
    \begin{align*}
        \frac{1}{p^*}\sum_{j=2}^p\E(Y_{x,j}^2|\mathcal{F}_{j-1})= \frac{1}{p^*}\sum_{j=2}^p\E(I_{x,j}^2|\mathcal{F}_{j-1})- \frac{1}{p^*}\sum_{j=2}^p(j-1)^2\{A_{1,n}(x)\}^2,
    \end{align*}
    which implies 
    \begin{align*}
        \var\left\{\frac{1}{p^*}\sum_{j=2}^p\E(Y_{x,j}^2|\mathcal{F}_{j-1})\right\}=\var\left\{\frac{1}{p^*}\sum_{j=2}^p\E(I_{x,j}^2|\mathcal{F}_{j-1})\right\}.
    \end{align*}
    By the convexity of the function $g(x)=x^2,\forall x\in\mR$, we have
    \begin{align}\label{eql8}
        \var\left\{\frac{1}{p^*}\sum_{j=2}^p\E(I_{x,j}^2|\mathcal{F}_{j-1})\right\}\le\frac{p-1}{(p^*)^2}\sum_{j=2}^p\var\left\{\E(I_{x,j}^2|\mathcal{F}_{j-1})\right\}.
    \end{align}
    By \eqref{eql6},
    \begin{align*}
        &\var\left\{\E(I_{x,j}^2|\mathcal{F}_{j-1})\right\}\\
        =&4\var\left[\sum_{1\le i_1<i_2\le j-1}\E\{h_x(\sqrt{n}\s_{i_1}^\top\s_j)h_x(\sqrt{n}\s_{i_2}^\top\s_j)|\s_{i_1},\s_{i_2}\}\right]\\
        =&4\sum_{1\le i_1<i_2\le j-1}\var[\E\{h_x(\sqrt{n}\s_{i_1}^\top\s_j)h_x(\sqrt{n}\s_{i_2}^\top\s_j)|\s_{i_1},\s_{i_2}\}]\\
        &+4\sum_{\substack{(i_1,i_2)\ne(i_3,i_4)\\1\le i_1<i_2\le j-1\\1\le i_3<i_4\le j-1}}\cov[\E\{h_x(\sqrt{n}\s_{i_1}^\top\s_j)h_x(\sqrt{n}\s_{i_2}^\top\s_j)|\s_{i_1},\s_{i_2}\}, \E\{h_x(\sqrt{n}\s_{i_3}^\top\s_j)h_x(\sqrt{n}\s_{i_4}^\top\s_j)|\s_{i_3},\s_{i_4}\}]\\
        =&4\sum_{1\le i_1<i_2\le j-1}\var[\E\{h_x(\sqrt{n}\s_{i_1}^\top\s_j)h_x(\sqrt{n}\s_{i_2}^\top\s_j)|\s_{i_1},\s_{i_2}\}],
    \end{align*}
    where the last step is due to \eqref{eql3}. Next, we analyze $\var[\E\{h_x(\sqrt{n}\s_1^\top\s_3)h_x(\sqrt{n}\s_2^\top\s_3)|\s_1,\s_2\}]$. Let $r=\s_1^\top\s_2$, and pick an $(n-1)\times(n-1)$ orthogonal matrix $\Q$ such that $\Q\s_1=\e_1$ and $\Q\s_2=r\e_1+\sqrt{1-r^2}\e_2$, where $\e_k$ is a $(n-1)$-dimensional vector with the $k$-th entry being one and the others being zero. Let $\u=(u_1,\ldots,u_{n-1})^\top:=\Q\s_3$. Then, $\u$ is uniformly on $\mS^{n-2}$ due to the rotational invariance, and 
    \begin{align*}
        \s_1^\top\s_3=u_1,~~\text{and}~~\s_2^\top\s_3=ru_1+\sqrt{1-r^2}u_2.
    \end{align*}
    Hence,
    \begin{align*}
        \E\{h_x(\sqrt{n}\s_1^\top\s_3)h_x(\sqrt{n}\s_2^\top\s_3)|\s_1,\s_2\}=\E_{u_1,u_2}\{h_x(\sqrt{n}u_1)h_x(r\sqrt{n}u_1+\sqrt{1-r^2}\sqrt{n}u_2)\}=:\phi(r).
    \end{align*}
    Since $u_2$ is symmetry, $\phi(r)=\phi(-r)$. By CLT, $\sqrt{n}r\cd N(0,1)$, then $r$ concentrates around $0$ with variance $n^{-1}$. Therefore, the Taylor expansion of $\phi(r)$ has the following form:
    \begin{align*}
        \phi(r)=\phi(0)+\frac{\phi^{(2)}(0)}{2!}r^2+\frac{\phi^{(4)}(0)}{4!}r^4+\cdots.
    \end{align*}
    To proceed,
    \begin{align*}
        \phi^{(2)}(0)=\E\{(\sqrt{n}u_1)^2h_x(\sqrt{n}u_1)h_x^{(2)}(\sqrt{n}u_2)-h_x(\sqrt{n}u_1)(\sqrt{n}u_2)h_x^{(1)}(\sqrt{n}u_2)\}.
    \end{align*}
    Note that $\sqrt{n}(u_1,u_2)\cd N(\bzero,\I_2)$ as $n\to\infty$. Then, by Stein's Lemma, we have
    \begin{align*}
        \phi^{(2)}(0)\approx\E\{Z^2h_x(Z)\}\E\{h_x^{(2)}(Z)\}-\E\{h_x(Z)\}\E\{Zh_x^{(1)}(Z)\}=[\E\{h_x^{(2)}(Z)\}]^2,
    \end{align*}
    where $Z\sim N(0,1)$. Hence, $\phi^{(2)}(0)<\infty$. Furthermore,
    \begin{align*}
        \var[\E\{h_x(\sqrt{n}\s_1^\top\s_3)h_x(\sqrt{n}\s_2^\top\s_3)|\s_1,\s_2\}]=\var\{\phi(r)\}\le C\var(r^2)=O(n^{-2}),
    \end{align*}
    where the last step follows from the fact the $r^2=(\s_1^\top\s_2)^2\sim\text{Beta}\left(\frac{1}{2},\frac{n-2}{2}\right)$ by \eqref{eql1}. Therefore,
    \begin{align*}
        \var\left\{\E(I_{x,j}^2|\mathcal{F}_{j-1})\right\}\le Cj^2n^{-2},
    \end{align*}
    which, joint with \eqref{eql8}, implies
    \begin{align}\label{eql9}
        \var\left\{\frac{1}{p^*}\sum_{j=2}^p\E(Y_{x,j}^2|\mathcal{F}_{j-1})\right\}=\var\left\{\frac{1}{p^*}\sum_{j=2}^p\E(I_{x,j}^2|\mathcal{F}_{j-1})\right\}=O(n^{-2}).
    \end{align}
    Combining with \eqref{eql7} and \eqref{eql9}, \eqref{eql4} is proved.

    \noindent\textit{Step 2.} We prove \eqref{eql5}. It suffices to show that 
    \begin{align*}
        \frac{1}{(p^*)^2}\sum_{j=2}^p\E(Y_{x,j}^4)\to 0.
    \end{align*}
    By \eqref{eql3}, using simple calculation, we have
    \begin{align*}
        \E(I_{x,j})=&(j-1)A_{1,n}(x),\\
        \E(I_{x,j}^2)=&(j-1)A_{2,n}(x)+(j-1)(j-2)\{A_{1,n}(x)\}^2,\\
        \E(I_{x,j}^3)=&(j-1)A_{3,n}(x)+3(j-1)(j-2)A_{2,n}(x)A_{1,n}(x)+6\binom{j-1}{3}\{A_{1,n}(x)\}^3,\\
        \E(I_{x,j}^4)=&(j-1)A_{4,n}(x)+4(j-1)(j-2)A_{3,n}(x)A_{1,n}(x)+6\binom{j-1}{2}\{A_{2,n}(x)\}^2\\
        &+12(j-1)\binom{j-2}{2}A_{2,n}(x)\{A_{1,n}(x)\}^2+24\binom{j-1}{4}\{A_{1,n}(x)\}^4.
    \end{align*}
    Hence, by Lemma \ref{lemma9},
    \begin{align*}
        \E(Y_{x,j}^4)=&\E[\{I_{x,j}-\E(I_{x,j})\}^4]\\
        =&(j-1)A_{4,n}(x)-4(j-1)A_{3,n}(x)A_{1,n}(x)+3(j-1)(j-2)\{A_{2,n}(x)\}^2\\
        &-6(j-1)(j-3)A_{2,n}(x)\{A_{1,n}(x)\}^2+3(j-1)(j-3)\{A_{1,n}(x)\}^4\\
        \to&(j-1)A_{4}(x)-4(j-1)A_{3}(x)A_{1}(x)+3(j-1)(j-2)\{A_{2}(x)\}^2\\
        &-6(j-1)(j-3)A_{2}(x)\{A_{1}(x)\}^2+3(j-1)(j-3)\{A_{1}(x)\}^4\\
        \le&Cj^2,
    \end{align*}
    which implies 
    \begin{align*}
        \frac{1}{(p^*)^2}\sum_{j=2}^p\E(Y_{x,j}^4)=O(p^{-1}).
    \end{align*}
    The proof of \eqref{eql5} is completed.
    
\end{proof}

\begin{lemma}\label{lemma11}
	Let $\{(U,U_N,\widetilde{U}_N)\in\mathbb{R}^3;N\geq 1\}$ and $\{(V,V_N,\widetilde{V}_N)\in\mathbb{R}^3;N\geq 1\}$ be two sequences of random variables with $U_N\stackrel{d}{\to} U$ and $V_N\stackrel{d}{\to} V$ as $N\to\infty$. Assume $U$ and $V$ are continuous random variables. We assume that
	\begin{equation*}
	\widetilde{U}_N=U_N+o_p(1)~~\text{and}~~ \widetilde{V}_N=V_N+o_p(1).
	\end{equation*}
	If $U_N$ and $V_N$ are asymptotically independent, then $\widetilde{U}_N$ and $\widetilde{V}_N$ are also asymptotically independent.
\end{lemma}
\begin{proof}
	See Lemma 7.10 in \cite{feng2024asymptotic}.
\end{proof}

\begin{lemma}\label{lemma12}
	For events $A_1,\ldots,A_p$, let $X=\sum_{i=1}^p\indicate(A_i)$, $S_j=\sum_{1\le i_1<\ldots<i_j\le p}P(A_{i_1}\cap\ldots\cap A_{i_j})$, then for any integer $k\in\{1,\ldots,p\}$, for any integer $1\le d\le [(p-k)/2]$,
	\begin{align*}
	\sum_{j=k}^{2d+k+1}(-1)^{k+j}\binom{j-1}{k-1}S_j \le \Pr(X\ge k) \le \sum_{j=k}^{2d+k}(-1)^{k+j}\binom{j-1}{k-1}S_j.
	\end{align*}
\end{lemma}
\begin{proof}
	See Chapter \uppercase\expandafter{\romannumeral4} in \cite{Feller1968Probability}.
\end{proof}

\begin{lemma}\label{lemma13}
    Let $\s_1,\ldots,\s_p$ be i.i.d. random vectors uniformly distributed on $\mS^{n-2}$. Set $\Lambda_p:=\{(i,j):1\le i<j\le p\}=\{I_k=(i_k,j_k):1\le k\le p^*\}$ with
    \begin{align*}
        I_k:=\left(k-\sum_{l=1}^{i-1}(p-l),k-\sum_{l=1}^{i-1}(p-l)+i\right),~~\text{if}~~\sum_{l=1}^{i-1}(p-l)<k\le \sum_{l=1}^{i}(p-l)~~\text{for some}~~1\le i\le p-1.
    \end{align*}
    For any $y\in\mR$ and $I_k=(i_k,j_k)\in\Lambda_p$, define $B_k:=\{n(\s_{i_k}^\top\s_{j_k})^2>4\log p-\log(\log p)+y\}$, and
    \begin{align*}
        H(p,l):=\sum_{1\le k_1<\cdots<k_l\le p^*}\Pr(B_{k_1}\cdots B_{k_l}).
    \end{align*}
    Assume $\log p=o(\sqrt{n})$, then 
    \begin{align*}
        \limsup_{p\to\infty}H(p,l)\le \frac{C}{l!}~~\text{for each}~~l\ge 3.
    \end{align*}
\end{lemma}
\begin{proof}
    See Lemma 36 and Equation (108) in \cite{feng2022max}.
\end{proof}

\begin{lemma}\label{lemma14}
    Let $m\ge 1$ and $\{\xi_i;1\le i\le m\}$ be independent random variables with $\E(\xi_1)=0$ for each $i$ and $\sup_{1\le i\le m}\E(|\xi_i|^\tau)<\infty$ for some $\tau\ge 2$. Then there exists a constant $C_{\tau}>0$ depending on $\tau$ only such that
    \begin{align*}
        \E\left(\left|\sum_{i=1}^m\xi_i\right|^\tau\right)\le C_\tau m^{\frac{\tau}{2}-1}\sum_{i=1}^m\E(|\xi_i|^\tau).
    \end{align*}
\end{lemma}
\begin{proof}
    See Lemma 2 in \cite{feng2022max}.
\end{proof}

\begin{lemma}\label{lemma15}
    Let $\xi$ be a random variable with $\E(\xi)=a$. Let $\tau\ge 2$ be given. The following holds.\\
    (\romannumeral1) If $a=0$, then
    \begin{align*}
        \E\{|\xi^2-\E(\xi^2)|^\tau\}\le32^\tau\E(|\xi|^{2\tau}).
    \end{align*}
    (\romannumeral2) If $a\ne 0$, then 
    \begin{align*}
        \E\{|\xi^2-\E(\xi^2)|^\tau\}\le16^\tau\left[|a|^{-\tau}\{\var(\xi)\}^\tau+\sqrt{\E(|\xi-a|^{2\tau})}\right]\cdot\left[|a|^{\tau}+|a|^{-\tau}\{\var(\xi)\}^\tau+\sqrt{\E(|\xi-a|^{2\tau})}\right].
    \end{align*}
\end{lemma}
\begin{proof}
    See Lemma 1 in \cite{feng2022max}.
\end{proof}

\subsection{Proof of Theorem \ref{asymptotic distribution of ordered statistics} (\romannumeral1)}\label{appendix2}
\begin{proof}
    Set $\Lambda_p:=\{(i,j):1\le i<j\le p\}=\{I_k=(i_k,j_k):1\le k\le p^*\}$ with
    \begin{align}\label{eq1}
        I_k:=\left(k-\sum_{l=1}^{i-1}(p-l),k-\sum_{l=1}^{i-1}(p-l)+i\right),~~\text{if}~~\sum_{l=1}^{i-1}(p-l)<k\le \sum_{l=1}^{i}(p-l)~~\text{for some}~~1\le i\le p-1.
    \end{align}
    For $I_k=(i_k,j_k)\in\Lambda_p$, write $\hat{\rho}_{I_k}=\hat{\rho}_{i_kj_k}$. Define $N_{p^*}(B)=\sum_{k=1}^{p^*}\indicate(n\hat{\rho}^2_{I_k}>x+b_p,k/p^*\in B)$ for any Borel set $B$ on $(0,1]$. Then,
	\begin{align*}
	\{n\hat{\rho}^2_{(p^*+1-s)}\le x+b_p\}=\{N_{p^*}((0,1])\le s-1\}.
	\end{align*}
	Let $N$ be a Poisson
	point process on $(0,1]$ with intensity $\log \Lambda^{-1}(x)$, then
	\begin{align*}
	\Pr\left(N((0,1])\le s-1\right)=\Lambda(x)\sum_{i=0}^{s-1}\frac{\{\log \Lambda^{-1}(x)\}^i}{i!}.
	\end{align*}
	Hence, it suffices to show that $N_{p^*}\cd N$. By Lemma \ref{lemma1}, it is sufficient to show that
	\begin{align}\label{eq2}
	    \E\{N_{p^*}((a,b])\} \to \E\{N((a,b])\} =(b-a)\log \Lambda^{-1}(x),0<a<b\le 1,
	\end{align}
	and for $0<a_1<b_1\le a_2<b_2\le\ldots\le a_l<b_l\le 1$, $l\ge 1$,
	\begin{align}\label{eq3}
	    \Pr\left(\bigcap_{i=1}^l\{N_{p^*}((a_i,b_i])=0\} \right)\to& \Pr\left(\bigcap_{i=1}^l\{N((a_i,b_i])=0\} \right)\nonumber\\
	    =&\exp\left[-\sum_{i=1}^l(b_i-a_i)\log \Lambda^{-1}(x)\right].
	\end{align}

    \noindent\textit{Step 1.} We prove \eqref{eq2}. Note that under $H_0$, $\bms$ is diagonal. Then,
    \begin{align}\label{eq4}
        \hat{\rho}_{ij}=\frac{(\Z_{.i}-\bar{Z}_i)^\top(\Z_{.j}-\bar{Z}_j)}{\Vert\Z_{.i}-\bar{Z}_i\Vert_2\cdot\Vert\Z_{.j}-\bar{Z}_j\Vert_2},
    \end{align}
    where $\Z_{.i}=(Z_{1i},\ldots,Z_{ni})^\top$ and $\bar{Z}_i=n^{-1}\sum_{k=1}^nZ_{ki}$. Here, we write $\Z_{.i}-\bar{Z}_i$ for $\Z_{.i}-\bar{Z}_i\mathbbm{1}_n$, where $\mathbbm{1}_n=(1,\ldots,1)^\top\in\mR^n$. Since $Z_{ij}$'s are i.i.d., then $\{\hat{\rho}_{I_k},k\in \Lambda_p\}$ are identically distributed and
    \begin{align}\label{eq5}
        \E\{N_{p^*}((a,b])\}=\sum_{k\in(ap^*,bp^*]}\Pr\left(n\hat{\rho}^2_{I_k}>x+b_p\right)=(b-a)p^*\Pr\left(n\hat{\rho}^2_{12}>x+b_p\right).
    \end{align}
    Furthermore, setting $\check{\rho}_{k},1\le k\le p^*$ be i.i.d. copies of $\hat{\rho}_{12}$, by Lemma \ref{lemma2},
    \begin{align*}
        \left[1-\left\{1-\Pr\left(n\hat{\rho}^2_{12}\le x+b_p\right)\right\}\right]^{p^*}=\left\{\Pr\left(n\hat{\rho}^2_{12}\le x+b_p\right)\right\}^{p^*}=\Pr\left(\max_{1\le k\le p^*}n\check{\rho}^2_k-b_p\le x\right)\to \Lambda(x).
    \end{align*}
    By taking logarithms, we have
    \begin{align*}
        -p^*\left\{1-\Pr\left(n\hat{\rho}^2_{12}\le x+b_p\right)\right\}\{1+o(1)\}\to\log\Lambda(x),
    \end{align*}
    which, combining with \eqref{eq5}, implies \eqref{eq2}.

    \noindent\textit{Step 2.} We prove \eqref{eq3}. Set $E_l:=\cup_{i=1}^l\{\lfloor a_ip^*\rfloor+1,\lfloor a_ip^*\rfloor+2,\ldots,\lfloor b_ip^*\rfloor\}$. For any $I_k=(i_k,j_k)\in\Lambda_p$ defined in \eqref{eq1}, set 
    \begin{align}\label{eq6}
        \tilde{\rho}_{I_k}:=\frac{1}{n}\Z_{.i_k}^\top\Z_{.j_k}.
    \end{align}
    Then,
    \begin{align}\label{eq7}
        &\Pr\left(\bigcap_{i=1}^l\{N_{p^*}((a_i,b_i])=0\} \right)\nonumber\\
        =&\Pr\left(\max_{k\in E_l}n\hat{\rho}^2_{I_k}\le x+b_p\right)\nonumber\\
        =&\Pr\left(\max_{k\in E_l}n\tilde{\rho}^2_{I_k}\le x+b_p\right)+\left\{\Pr\left(\max_{k\in E_l}n\hat{\rho}^2_{I_k}\le x+b_p\right)-\Pr\left(\max_{k\in E_l}n\tilde{\rho}^2_{I_k}\le x+b_p\right)\right\}.
    \end{align}

    \noindent\textit{Step 2.1} We first show that 
    \begin{align}\label{eq8}
        \Pr\left(\max_{k\in E_l}n\tilde{\rho}^2_{I_k}\le x+b_p\right)\to \exp\left[-\sum_{i=1}^l(b_i-a_i)\log \Lambda^{-1}(x)\right].
    \end{align}
    For any $I_k=(i_k,j_k)\in\Lambda_p$ defined in \eqref{eq1}, define 
    \begin{align*}
        B_{I_k}:=\{I_{k'}\in\Lambda_p:k'\in E_l,|I_k\cap I_{k'}|\ge 1,I_k\ne I_{k'}\}.
    \end{align*}
    Then, by Lemma \ref{lemma3},
    \begin{align}\label{eq9}
        \left|\Pr\left(\max_{k\in E_l}n\tilde{\rho}^2_{I_k}\le x+b_p\right)-\exp(-\lambda_{1,n})\right|\le 2(b_{1,n}+b_{2,n}),
    \end{align}
    where, 
    \begin{align*}
        \lambda_{1,n}:=&\sum_{k\in E_l}\Pr\left(n\tilde{\rho}^2_{I_k}>x+b_p\right)=\sum_{i=1}^l(b_i-a_i)p^*\Pr\left(n\tilde{\rho}^2_{I_1}>x+b_p\right),\\
        b_{1,n}:=&\sum_{I_k:k\in E_l}\sum_{I_{k'}\in B_{I_k}}\Pr\left(n\tilde{\rho}^2_{I_k}>x+b_p\right)\Pr\left(n\tilde{\rho}^2_{I_k'}>x+b_p\right)\le Cp^3\left\{\Pr\left(n\tilde{\rho}^2_{I_1}>x+b_p\right)\right\}^2,~~\text{and}\\
        b_{2,n}:=&\sum_{I_k:k\in E_l}\sum_{I_{k'}\in B_{I_k}}\Pr\left(n\tilde{\rho}^2_{I_k}>x+b_p,n\tilde{\rho}^2_{I_k'}>x+b_p\right)\le Cp^3\Pr\left(n\tilde{\rho}^2_{I_1}>x+b_p,n\tilde{\rho}^2_{I_2}>x+b_p\right),
    \end{align*}
    from the fact that $\{\tilde{\rho}_{I_k},k\in \Lambda_p\}$ are identically distributed. We first calculate $\lambda_{1,n}$. Recalling the definition \eqref{eq6},
    \begin{align*}
        \Pr\left(n\tilde{\rho}^2_{I_1}>x+b_p\right)=\Pr\left(\frac{1}{\sqrt{n\log p}}\left|\sum_{k=1}^nZ_{k1}Z_{k2}\right|>\sqrt{\frac{x+b_p}{\log p}}\right).
    \end{align*}
    Note that $\sqrt{(x+b_p)/\log p}\to 2$ as $p\to\infty$, and $\{Z_{k1}Z_{k2},1\le k\le n\}$ are i.i.d. with $\E(Z_{k1}Z_{k2})=0$, $\E\{(Z_{k1}Z_{k2})^2\}=1$ and for $0<a\le 2$,
    \begin{align*}
        |Z_{k1}Z_{k2}|^{a/2}\le\left(\frac{Z_{k1}^2+Z_{k2}^2}{2}\right)^{a/2}\le \frac{1}{2^{a/2}}\left(|Z_{k1}|^a+|Z_{k2}|^a\right).
    \end{align*}
    Hence, by independence, 
    \begin{align*}
        \E\{\exp(t_0|Z_{k1}Z_{k2}|^{a/2})\}<\infty.
    \end{align*}
    Thus, by Lemma \ref{lemma4}, we get
    \begin{align*}
        \Pr\left(\frac{1}{\sqrt{n\log p}}\left|\sum_{k=1}^nZ_{k1}Z_{k2}\right|>\sqrt{\frac{x+b_p}{\log p}}\right)
        \sim\frac{p^{-\frac{x+b_p}{2\log p}}(\log p)^{-1/2}}{\sqrt{8\pi}}
        \sim\frac{1}{p^2}\cdot\frac{\exp(-x/2)}{\sqrt{8\pi}},
    \end{align*}
    as $\min(n,p)\to\infty$. Considering $\E(Z_{ki})=0$, it is easy to see that the above also holds if $Z_{k1}Z_{k2}$ is replaced by $-Z_{k1}Z_{k2}$. These lead to
    \begin{align*}
        \Pr\left(n\tilde{\rho}^2_{I_1}>x+b_p\right)\sim\frac{2}{p^2}\cdot\frac{\exp(-x/2)}{\sqrt{8\pi}},
    \end{align*}
    and 
    \begin{align}\label{eq10}
        \lambda_{1,n}\sim\sum_{i=1}^l(b_i-a_i)\cdot\frac{\exp(-x/2)}{\sqrt{8\pi}},
    \end{align}
    as $\min(n,p)\to\infty$. Furthermore,
    \begin{align}\label{eq11}
        b_{1,n}\le Cp^3\left\{\frac{\lambda_{1,n}}{\sum_{i=1}^l(b_i-a_i)p^*}\right\}^2=O(p^{-1}).
    \end{align}
    To proceed, by Lemma \ref{lemma5}, for any $\varepsilon>0$,
    \begin{align}\label{eq12}
        b_{2,n}\le& Cp^3\Pr\left(n\tilde{\rho}^2_{I_1}>x+b_p,n\tilde{\rho}^2_{I_2}>x+b_p\right)\nonumber\\
        =&Cp^3\Pr\left(|\Z_{.1}^\top\Z_{.2}|>\sqrt{n(x+b_p)},|\Z_{.2}^\top\Z_{.3}|>\sqrt{n(x+b_p)}\right)\nonumber\\
        =&Cp^3\E\left[\left\{\Pr\left(|\Z_{.1}^\top\Z_{.2}|>\sqrt{n(x+b_p)}|\Z_{.2}\right)\right\}^2\right]=O(p^3\cdot p^{-4+\varepsilon}).
    \end{align}
    Combining \eqref{eq9}-\eqref{eq12} implies \eqref{eq8}.

    \noindent\textit{Step 2.2} Next we show that 
    \begin{align}\label{eq13}
        \left|\max_{k\in E_l}n\hat{\rho}^2_{I_k}-\max_{k\in E_l}n\tilde{\rho}^2_{I_k}\right|=o_p(1).
    \end{align}
    Observe that
    \begin{align}\label{eq14}
        \left|\max_{k\in E_l}n\hat{\rho}^2_{I_k}-\max_{k\in E_l}n\tilde{\rho}^2_{I_k}\right|=&n^{-1}\left|\max_{k\in E_l}n|\hat{\rho}_{I_k}|-\max_{k\in E_l}n|\tilde{\rho}_{I_k}|\right|\cdot\left|\max_{k\in E_l}n|\hat{\rho}_{I_k}|+\max_{k\in E_l}n|\tilde{\rho}_{I_k}|\right|\nonumber\\
        \le&n^{-1}\Delta_n(\Delta_n+2W_n),
    \end{align}
    where $\Delta_n:=\max_{k\in E_l}|n\hat{\rho}_{I_k}-n\tilde{\rho}_{I_k}|$ and $W_n:=\max_{k\in E_l}|n\tilde{\rho}_{I_k}|$. By \eqref{eq8}, setting $x=\log(\log p)$, we get
    \begin{align}\label{eq15}
        \frac{W_n}{\sqrt{n\log p}}\cp 2.
    \end{align}
    For $I_k=(i_k,j_k)\in\Lambda_p$,
    \begin{align}
        |n\hat{\rho}_{I_k}-n\tilde{\rho}_{I_k}|=&\left|\frac{\Z_{.i_k}^\top\Z_{.j_k}-n\bar{Z}_{i_k}\bar{Z}_{j_k}}{h_{i_k}h_{j_k}}-\Z_{.i_k}^\top\Z_{.j_k}\right|\nonumber\\
        =&|(h_{i_k}h_{j_k})^{-1}-1|\cdot|\Z_{.i_k}^\top\Z_{.j_k}|+n\frac{|\bar{Z}_{i_k}|\cdot|\bar{Z}_{j_k}|}{h_{i_k}h_{j_k}}\nonumber\\
        =&|\{(h_{i_k}-1)(h_{j_k}-1)+(h_{i_k}-1)+(h_{j_k}-1)\}(h_{i_k}h_{j_k})^{-1}|\cdot|\Z_{.i_k}^\top\Z_{.j_k}|+n\frac{|\bar{Z}_{i_k}|\cdot|\bar{Z}_{j_k}|}{h_{i_k}h_{j_k}},
    \end{align}
    where $h_i:=\Vert\Z_{.i}-\bar{Z}_i\Vert_2/\sqrt{n}$. Taking maximum for both side, we obtain
    \begin{align*}
        \Delta_n\le(b_{n,1}^2+2b_{n,1})b_{n,3}^{-2}W_n+nb_{n,3}^{-2}b_{n,4}^2,
    \end{align*}
    where
    \begin{align*}
        b_{n,1}:=\max_{1\le i\le p}|h_i-1|,~~b_{n,3}:=\min_{1\le i\le p}h_i~~\text{and}~~b_{n,4}:=\max_{1\le i\le p}|\bar{Z}_i|.
    \end{align*} 
    It follows that
    \begin{align*}
        \frac{\Delta_n}{\log p}\le \left\{\sqrt{\frac{\log p}{n}}\left(\sqrt{\frac{n}{\log p}}b_{n,1}\right)^2+2\sqrt{\frac{n}{\log p}}b_{n,1}\right\}b_{n,3}^{-2}\cdot\frac{W_n}{\sqrt{n\log p}}+b_{n,3}^{-2}\left(\sqrt{\frac{n}{\log p}}b_{n,4}\right)^2,
    \end{align*}
    Then, by Lemma \ref{lemma6} and \eqref{eq15}, we obtain
    \begin{align*}
        \left\{\frac{\Delta_n}{\log p}\right\}~~\text{are tight},
    \end{align*}
    which, joint with \eqref{eq14} and \eqref{eq15}, implies
    \begin{align*}
        \left|\max_{k\in E_l}n\hat{\rho}^2_{I_k}-\max_{k\in E_l}n\tilde{\rho}^2_{I_k}\right|\le\sqrt{\frac{\log^3 p}{n}}\cdot\frac{\Delta_n}{\log p}\left(\sqrt{\frac{\log p}{n}}\cdot\frac{\Delta_n}{\log p}+2\frac{W_n}{\sqrt{n\log p}}\right)=o_p(1).
    \end{align*}
    That is \eqref{eq13}. Combining \eqref{eq7}, \eqref{eq8} and \eqref{eq13}, the proof of \eqref{eq3} is completed. 
    
    In summary, \eqref{eq2} and \eqref{eq3} are proved.
\end{proof}

\subsection{Proof of Theorem \ref{asymptotic distribution of ordered statistics} (\romannumeral2)}\label{appendix3}
\begin{proof}
    Recalling the definition \eqref{eq1}, for Borel sets $B\in(0,1]\times\mR$, define 
	\begin{align*}
	\widetilde{N}_{p^*}^{(s)}(B):=\sum_{k=1}^{p^*}\indicate(n\hat{\rho}^2_{I_k}>x_s+b_p,(k/p^*,s)\in B),~~\text{and}~~\widetilde{N}_{p^*}(B):= \sum_{s=1}^l\widetilde{N}_{p^*}^{(s)}(B).
	\end{align*}
	Then,
    \begin{align*}
        \Pr\left(\bigcap_{s=1}^l\left(n\hat{\rho}^2_{(p^*+1-s)}-b_p \leq x_s\right)\right)=\Pr\left(\bigcap_{s=1}^l\left\{ \widetilde{N}^{(s)}_{p^*}((0,1]\times (s-1/2,s+1/2] )\le s-1\right\}\right).
	\end{align*}
	Construct a point process $\widetilde{N}$ on $(0,1]\times\mR$ such that $\widetilde{N}(\cdot)=\sum_{s=1}^l\widetilde{N}^{(s)}(\cdot)$, where for each $1\le s\le l-1$, $\widetilde{N}^{(s)}$ is a Poisson process independent thinning of the Poisson process $\widetilde{N}^{(s+1)}$ with deleting probability $1-\frac{\log \Lambda^{-1}(x_s)}{\log \Lambda^{-1}(x_{s+1})}$, the initial Poisson process $\widetilde{N}^{(l)}$ has intensity $\log \Lambda^{-1}(x_l)$. According to Section 5 in \cite{Leadbetter1983ExtremesAR}, it is known that $\widetilde{N}^{(s)},s=1,\ldots,l$ are Poisson processes with intensity $\log \Lambda^{-1}(x_s)$ respectively, and are independent on disjoint intervals on the plane, and
	\begin{align*}
	    &\Pr\left( \bigcap_{s=1}^l\left\{ \widetilde{N}^{(s)}((0,1]\times (s-1/2,s+1/2] )\le s-1\right\} \right)\\
	    =& \Lambda\left(x_l\right) \sum_{\sum_{i=2}^s l_i \leq s-1, s=2, \ldots, l} \prod_{i=2}^l \frac{\{\log \Lambda^{-1}(x_i)-\log \Lambda^{-1}(x_{i-1})\}^{l_i}}{l_i!}.
	\end{align*}
	Hence, it suffices to show that $\widetilde{N}_{p^*}\cd\widetilde{N}$. By Lemma \ref{lemma1}, it is sufficient to show that for $0<a<b\le 1$, $0<c<d\le l$,
	\begin{align}\label{eq17}
	    \E\{\widetilde{N}_{p^*}((a,b]\times (c,d])\} \to \E\{\widetilde{N}((a,b]\times (c,d])\} =(b-a)\sum_{c<s\le d}\log \Lambda^{-1}(x_s),
	\end{align}
	and for $0<a_1<b_1\le a_2<b_2\le\ldots\le a_s<b_s\le 1$, $0<c_1<d_1\le c_2<d_2\le\ldots\le c_s<d_s\le l$,
	\begin{align}\label{eq18}
	    \Pr\left(\bigcap_{s=1}^l\{\widetilde{N}_{p^*}((a_s,b_s]\times (c_l,d_l])=0\} \right)&\to \Pr\left(\bigcap_{s=1}^l\{\widetilde{N}((a_s,b_s]\times (c_s,d_s])=0\} \right)\nonumber\\
	    &=\exp\left[-\sum_{s=1}^l(b_s-a_s)\log \Lambda^{-1}(x_{d_s})\right].
	\end{align}

    For \eqref{eq17}, by \eqref{eq2},
    \begin{align*}
        \E\{\widetilde{N}_{p^*}((a,b]\times (c,d])\}=\sum_{s\in(c,d]}\sum_{k\in(ap^*,bp^*]}\Pr(n\hat{\rho}^2_{I_k}>x_s+b_p)\to \sum_{c<s\le d}(b-a)\log \Lambda^{-1}(x_s).
    \end{align*}

    For \eqref{eq18}, by \eqref{eq13},
    \begin{align*}
        \Pr\left(\bigcap_{s=1}^l\{\widetilde{N}_{p^*}((a_s,b_s]\times (c_l,d_l])=0\} \right)=&\Pr\left(\bigcap_{s=1}^l\left\{\max_{k\in(a_lp^*,b_lp^*]}n\hat{\rho}^2_{I_k}\le x_{d_l}+b_p\right\} \right)\\
        =&\Pr\left(\bigcap_{s=1}^l\left\{\max_{k\in(a_lp^*,b_lp^*]}n\tilde{\rho}^2_{I_k}\le x_{d_l}+b_p\right\} \right)+o(1).
    \end{align*}
    Using Lemma \ref{lemma3} and the similar arguments as thoes used in the proof of \eqref{eq8}, we can prove
    \begin{align*}
        \Pr\left(\bigcap_{s=1}^l\left\{\max_{k\in(a_lp^*,b_lp^*]}n\tilde{\rho}^2_{I_k}\le x_{d_l}+b_p\right\} \right)\to \exp\left[-\sum_{s=1}^l(b_s-a_s)\log \Lambda^{-1}(x_{d_s})\right].
    \end{align*}
\end{proof}

\subsection{Proof of Theorem \ref{bootstrap}}\label{appendix4}
\begin{proof}
    Recalling \eqref{eq4}, 
    \begin{align*}
        \hat{\rho}_{ij}=\frac{(\Z_{.i}-\bar{Z}_i)^\top(\Z_{.j}-\bar{Z}_j)}{\Vert\Z_{.i}-\bar{Z}_i\Vert_2\cdot\Vert\Z_{.j}-\bar{Z}_j\Vert_2}=\frac{\Z_{.i}^\top\Z_{.j}-n\bar{Z}_i\bar{Z}_j}{\Vert\Z_{.i}-\bar{Z}_i\Vert_2\cdot\Vert\Z_{.j}-\bar{Z}_j\Vert_2}.
    \end{align*}
    For each $1\le i\le p$, we sample a permutation $\pi_i=\{\pi_i(1),\ldots,\pi_i(n)\}$ of $\{1,\ldots,n\}$ uniformly at random. Then, under $H_0$, we can write
    \begin{align*}
        \X_{.i}^*=\mu_i\e+\sigma_{ii}\Z_{.i}^*,
    \end{align*}
    where $\X_{.i}^*=(X_{\pi_i(1),i},\ldots,X_{\pi_i(n),i})^{\top}$ and $\Z_{.i}^*=(Z_{\pi_i(1),i},\ldots,Z_{\pi_i(n),i})^{\top}$. Hence,
    \begin{align*}
        \hat{\rho}^*_{ij}=\frac{(\Z^*_{.i}-\bar{Z}^*_i)^\top(\Z^*_{.j}-\bar{Z}^*_j)}{\Vert\Z^*_{.i}-\bar{Z}^*_i\Vert_2\cdot\Vert\Z^*_{.j}-\bar{Z}^*_j\Vert_2}=\frac{(\Z^*_{.i})^\top\Z^*_{.j}-n\bar{Z}_i\bar{Z}_j}{\Vert\Z_{.i}-\bar{Z}_i\Vert_2\cdot\Vert\Z_{.j}-\bar{Z}_j\Vert_2}.
    \end{align*}
    Since $\{Z_{ij}:1\le i\le n,1\le j\le p\}$ are i.i.d., $(\Z^*_{.i})^\top\Z^*_{.j}\stackrel{d}{=}\Z_{.i}^\top\Z_{.j}$, which implies $\hat{\rho}^*_{ij}\stackrel{d}{=}\hat{\rho}_{ij}$. Thus, the observations are exchangeable.   
\end{proof}

\subsection{Proof of Theorem \ref{asymptotic normality}}\label{appendix5}
\begin{proof}
    By \eqref{eq4} and Lemma \ref{lemma7}, we have  
    \begin{align}\label{eq19}
        \hat{\rho}_{ij}=\frac{(\P\Z_{.i})^\top(\P\Z_{.j})}{\Vert\P\Z_{.i}\Vert_2\cdot\Vert\P\Z_{.j}\Vert_2}=\s_i^\top\s_j,~~1\le i<j\le p,
    \end{align}
    where $\P=\I_n-n^{-1}\bone_n\bone_n^\top$ and $\s_1,\ldots,\s_p$ are i.i.d. random vectors uniformly distributed on $\mS^{n-2}$. Then,
    \begin{align*}
        \frac{1}{\sqrt{p^*}}T_{\lceil\gamma p^*\rceil}=\frac{1}{\sqrt{p^*}}\sum_{1\le i<j\le p}n(\s_i^\top\s_j)^2\indicate\{n(\s_i^\top\s_j)^2\ge\hat{v}_{\gamma,n,p}\},
    \end{align*}
    where 
    \begin{align*}
        \hat{v}_{\gamma,n,p}:=\inf\{x:\hat{F}_{n,p}(x)\ge 1-\gamma\}~~\text{with}~~\hat{F}_{n,p}(x):=\frac{2}{p(p-1)}\sum_{1\le i<j\le p}\indicate(n(\s_i^\top\s_j)^2\le x).
    \end{align*}    
    We first claim that
    \begin{align}\label{eq20}
        \frac{1}{\sqrt{p^*}}T_{\lceil\gamma p^*\rceil}=\frac{1}{\sqrt{p^*}}\sum_{j=2}^p\{I_{v_{\gamma,n},j}+(j-1)\gamma v_{\gamma,n}\}+o_p(1),
    \end{align}
    where
    \begin{align*}
        v_{\gamma,n}:=\inf\{x:F_n(x)\ge 1-\gamma\}~~\text{with}~~F_n(x):=\Pr(n(\s_1^\top\s_2)^2\le x),
    \end{align*}
    and for any $x\in\mR$,
    \begin{align*}
        I_{x,j}:=\sum_{i=1}^{j-1}\{n(\s_i^\top\s_j)^2-x\}\indicate\{n(\s_i^\top\s_j)^2\ge x\}.
    \end{align*}
    By Lemma \ref{lemma10}, as $\min(n,p)\to\infty$,
    \begin{align}\label{eq21}
        \frac{1}{\sqrt{p^*}}\sum_{j=2}^p\{I_{x,j}-\E(I_{x,j})\}\cd N(0,\sigma_{x,\gamma}^2),
    \end{align}
    where $\E(I_{x,j})=(j-1)A_{1,n}(x)$, $\sigma_{x,\gamma}^2:=A_2(x)-\{A_1(x)\}^2$ and
    \begin{align*}
        A_{1,n}(x)=&\frac{n}{n-1}\left\{1-F_{\text{Beta}\left(\frac{3}{2},\frac{n-2}{2}\right)}\left(\frac{x}{n}\right)\right\}-x\left\{1-F_{\text{Beta}\left(\frac{1}{2},\frac{n-2}{2}\right)}\left(\frac{x}{n}\right)\right\},\\
        A_{1}(x):=&\{1-F_{\chi^2_3}(x)\}-x\{1-F_{\chi^2_1}(x)\},\\
        A_{2}(x):=&3\{1-F_{\chi^2_5}(x)\}-2x\{1-F_{\chi^2_3}(x)\}+x^2\{1-F_{\chi^2_1}(x)\}.
    \end{align*}
    By CLT, $\sqrt{n}\s_1^\top\s_2\cd N(0,1)$, then
    \begin{align*}
        v_{\gamma,n}\to v_{\gamma},~~\text{as}~~n\to\infty,
    \end{align*}
    where $v_{\gamma}:=\inf\{x:F_{\chi^2_1}(x)\ge 1-\gamma\}$. This, joint with \eqref{eq20} and \eqref{eq21}, by Slutsky's lemma, implies
    \begin{align*}
        \frac{T_{\lceil\gamma p^*\rceil}-\mu_{\gamma,n,p}}{\sqrt{p^*}}\cd N(0,\sigma_{\gamma}^2),
    \end{align*}
    where $\mu_{\gamma,n,p}:=p^*\{A_{1,n}(v_{\gamma,n})+\gamma v_{\gamma,n}\}$ and $\sigma_{\gamma}^2:=A_2(v_\gamma)-\{A_1(v_\gamma)\}^2$.

    To complete the proof, we next prove the claim \eqref{eq20}. It is enough to show 
    \begin{align}\label{eq22}
        \frac{1}{\sqrt{p^*}}\sum_{1\le i<j\le p}\{n(\s_i^\top\s_j)^2-v_{\gamma,n}\}[\indicate\{n(\s_i^\top\s_j)^2\ge\hat{v}_{\gamma,n,p}\}-\indicate\{n(\s_i^\top\s_j)^2\ge v_{\gamma,n}\}]=o_p(1).
    \end{align}
    By Lemma \ref{lemma8}, we have for each $n$, 
    \begin{align}\label{eq23}
        \hat{v}_{\gamma,n,p}-v_{\gamma,n}=O_p(p^{-1}).
    \end{align}
    Let $W_{ij}:=\{n(\s_i^\top\s_j)^2-v_{\gamma,n}\}[\indicate\{n(\s_i^\top\s_j)^2\ge\hat{v}_{\gamma,n,p}\}-\indicate\{n(\s_i^\top\s_j)^2\ge v_{\gamma,n}\}]$. We first evaluate $\E(W_{ij})$. Note that, picking $a\in\left(\frac{1}{3},\frac{1}{2}\right)$ and $c>0$,
    \begin{align*}
        &\indicate\{n(\s_i^\top\s_j)^2\ge\hat{v}_{\gamma,n,p}\}-\indicate\{n(\s_i^\top\s_j)^2\ge v_{\gamma,n}\}\\
        =&-\indicate\{v_{\gamma,n}\le n(\s_i^\top\s_j)^2<\hat{v}_{\gamma,n,p}\}\indicate(\hat{v}_{\gamma,n,p}>v_{\gamma})+\indicate\{\hat{v}_{\gamma,n,p}\le n(\s_i^\top\s_j)^2<v_{\gamma,n}\}\indicate(\hat{v}_{\gamma,n,p}<v_{\gamma})\\
        =&-\indicate\{v_{\gamma,n}\le n(\s_i^\top\s_j)^2<\hat{v}_{\gamma,n,p}\}\indicate(\hat{v}_{\gamma,n,p}\ge v_{\gamma}+cp^{-2a})\\
        &-\indicate\{v_{\gamma,n}\le n(\s_i^\top\s_j)^2<\hat{v}_{\gamma,n,p}\}\indicate(v_{\gamma}<\hat{v}_{\gamma,n,p}<v_{\gamma}+cp^{-2a})\\
        &+\indicate\{\hat{v}_{\gamma,n,p}\le n(\s_i^\top\s_j)^2<v_{\gamma,n}\}\indicate(\hat{v}_{\gamma,n,p}\le v_{\gamma}-cp^{-2a})\\
        &+\indicate\{\hat{v}_{\gamma,n,p}\le n(\s_i^\top\s_j)^2<v_{\gamma,n}\}\indicate(v_{\gamma}-cp^{-2a}<\hat{v}_{\gamma,n,p}<v_{\gamma})\\
        =&:-I_{ij,11}-I_{ij,12}+I_{ij,21}+I_{ij,22}.
    \end{align*}
    Then $W_{ij}=\{n(\s_i^\top\s_j)^2-v_{\gamma,n}\}(-I_{ij,11}-I_{ij,12}+I_{ij,21}+I_{ij,22})$. Applying the Cauchy–Swartz inequality, by \eqref{eq23}, for $l=1,2$,
    \begin{align}\label{eq24}
        |\E[\{n(\s_i^\top\s_j)^2-v_{\gamma,n}\}I_{ij,l1}]|\le\sqrt{\E(|\hat{v}_{\gamma,n,p}-v_{\gamma,n}|^2)\Pr(|\hat{v}_{\gamma,n,p}-v_{\gamma,n}|\ge cp^{-2a})}=o(p^{-1}).
    \end{align}
    To proceed,
    \begin{align}\label{eq25}
        |\E[\{n(\s_i^\top\s_j)^2-v_{\gamma,n}\}I_{ij,12}]|\le&\E[\{n(\s_i^\top\s_j)^2-v_{\gamma,n}\}\indicate\{v_{\gamma,n}\le n(\s_i^\top\s_j)^2<v_{\gamma}+cp^{-2a})]\nonumber\\
        =&\int_{v_{\gamma,n}}^{v_{\gamma}+cp^{-2a}}(x-v_{\gamma,n})f_n(x)\dif x=O(p^{-4a}),
    \end{align}
    where $f_n(x)$ is the density of $n(\s_1^\top\s_2)^2$ given in \eqref{eql1}. Using exactly the same approach we can show that $|\E[\{n(\s_i^\top\s_j)^2-v_{\gamma,n}\}I_{ij,22}]|=O(p^{-4a})$ as well. Hence, $\E(W_{ij})=o(p^{-1})$, and
    \begin{align}\label{eq26}
        \frac{1}{\sqrt{p^*}}\sum_{1\le i<j\le p}\E(W_{ij})=o(1).
    \end{align}

    We now consider $\var(W_{ij})$. Note that
    \begin{align*}
        &[\indicate\{n(\s_i^\top\s_j)^2\ge\hat{v}_{\gamma,n,p}\}-\indicate\{n(\s_i^\top\s_j)^2\ge v_{\gamma,n}\}]^2\\
        =&\indicate\{n(\s_i^\top\s_j)^2\ge\hat{v}_{\gamma,n,p}\}-2\indicate\{n(\s_i^\top\s_j)^2\ge\hat{v}_{\gamma,n,p}\}\indicate\{n(\s_i^\top\s_j)^2\ge v_{\gamma,n}\}+\indicate\{n(\s_i^\top\s_j)^2\ge v_{\gamma,n}\}\\
        =&\indicate\{\hat{v}_{\gamma,n,p}\le n(\s_i^\top\s_j)^2<v_{\gamma,n}\}+\indicate\{v_{\gamma,n}\le n(\s_i^\top\s_j)^2<\hat{v}_{\gamma,n,p}\}\\
        =&\indicate\{\hat{v}_{\gamma,n,p}\le n(\s_i^\top\s_j)^2<v_{\gamma,n}\}\indicate\{\hat{v}_{\gamma,n,p}< v_{\gamma,n}-cp^{-2a}\}\\
        &+\indicate\{\hat{v}_{\gamma,n,p}\le n(\s_i^\top\s_j)^2<v_{\gamma,n}\}\indicate\{\hat{v}_{\gamma,n,p}\ge v_{\gamma,n}-cp^{-2a}\}\\
        &+\indicate\{v_{\gamma,n}\le n(\s_i^\top\s_j)^2<\hat{v}_{\gamma,n,p}\}\indicate\{\hat{v}_{\gamma,n,p}\ge v_{\gamma,n}+cp^{-2a}\}\\
        &+\indicate\{v_{\gamma,n}\le n(\s_i^\top\s_j)^2<\hat{v}_{\gamma,n,p}\}\indicate\{\hat{v}_{\gamma,n,p}<v_{\gamma,n}+cp^{-2a}\}\\
        =&:L_{ij,11}+L_{ij,12}+L_{ij,21}+L_{ij,22}.
    \end{align*}
    Then $\E(W_{ij}^2)=\E[\{n(\s_i^\top\s_j)^2-v_{\gamma,n}\}^2(L_{ij,11}+L_{ij,12}+L_{ij,21}+L_{ij,22})]$. Similar as the proof of \eqref{eq24} and \eqref{eq25}, we have for $l=1,2$,
    \begin{align*}
        &|\E[\{n(\s_i^\top\s_j)^2-v_{\gamma,n}\}^2L_{ij,l1}]|=o(p^{-2}),\\
        &|\E[\{n(\s_i^\top\s_j)^2-v_{\gamma,n}\}^2L_{ij,12}]|=O(p^{-6a}),
    \end{align*}
    which implies
    \begin{align}\label{eq27}
        \E(W_{ij}^2)=o(p^{-2}).
    \end{align}

    For $(i_1,j_1)\ne(i_2,j_2)$, we have
    \begin{align*}
        &(-I_{i_1j_1,11}-I_{i_1j_1,12}+I_{i_1j_1,21}+I_{i_1j_1,22})(-I_{i_2j_2,11}-I_{i_2j_2,12}+I_{i_2j_2,21}+I_{i_2j_2,22})\\
        =&(I_{i_1j_1,21}I_{i_2j_2,21}-I_{i_1j_1,21}I_{i_2j_2,11}-I_{i_1j_1,11}I_{i_2j_2,21}+I_{i_1j_1,11}I_{i_2j_2,11})\\
        &+(I_{i_1j_1,21}I_{i_2j_2,22}+I_{i_1j_1,22}I_{i_2j_2,21}-I_{i_1j_1,21}I_{i_2j_2,12}-I_{i_1j_1,22}I_{i_2j_2,11}-I_{i_1j_1,11}I_{i_2j_2,22}-I_{i_1j_1,12}I_{i_2j_2,21}\\
        &\quad+I_{i_1j_1,11}I_{i_2j_2,12}+I_{i_1j_1,12}I_{i_2j_2,11})\\
        &+(I_{i_1j_1,22}I_{i_2j_2,22}-I_{i_1j_1,22}I_{i_2j_2,12}-I_{i_1j_1,12}I_{i_2j_2,22}+I_{i_1j_1,12}I_{i_2j_2,12})\\
        =&:I_{i_1j_1i_2j_2,1}+I_{i_1j_1i_2j_2,2}+I_{i_1j_1i_2j_2,3}.
    \end{align*}
    Then, $\E(W_{i_1j_1}W_{i_2j_2})=\E[\{n(\s_{i_1}^\top\s_{j_1})^2-v_{\gamma,n}\}\{n(\s_{i_2}^\top\s_{j_2})^2-v_{\gamma,n}\}(I_{i_1j_1i_2j_2,1}+I_{i_1j_1i_2j_2,2}+I_{i_1j_1i_2j_2,3})]$. By the Cauchy–Swartz inequality and the similar results as \eqref{eq24} and \eqref{eq25}, we can bound the terms in $I_{i_1j_1i_2j_2,1}$ as 
    \begin{align*}
        &\E[\{n(\s_{i_1}^\top\s_{j_1})^2-v_{\gamma,n}\}\{n(\s_{i_2}^\top\s_{j_2})^2-v_{\gamma,n}\}I_{i_1j_1,21}I_{i_2j_2,21}]\\
        \le&\sqrt{\E[\{n(\s_{i_1}^\top\s_{j_1})^2-v_{\gamma,n}\}^2I_{i_1j_1,21}]\cdot \E[\{n(\s_{i_2}^\top\s_{j_2})^2-v_{\gamma,n}\}^2I_{i_2j_2,21}]}=o(p^{-2}),
    \end{align*}
    the terms in $I_{i_1j_1i_2j_2,2}$ as
    \begin{align*}
        &\E[\{n(\s_{i_1}^\top\s_{j_1})^2-v_{\gamma,n}\}\{n(\s_{i_2}^\top\s_{j_2})^2-v_{\gamma,n}\}I_{i_1j_1,21}I_{i_2j_2,22}]\\
        \le&\sqrt{\E[\{n(\s_{i_1}^\top\s_{j_1})^2-v_{\gamma,n}\}^2I_{i_1j_1,21}]\cdot \E[\{n(\s_{i_2}^\top\s_{j_2})^2-v_{\gamma,n}\}^2I_{i_2j_2,22}]}=o(p^{-3a-1}),
    \end{align*}
    and the terms in $I_{i_1j_1i_2j_2,3}$ as
    \begin{align*}
        &\E[\{n(\s_{i_1}^\top\s_{j_1})^2-v_{\gamma,n}\}\{n(\s_{i_2}^\top\s_{j_2})^2-v_{\gamma,n}\}I_{i_1j_1,22}I_{i_2j_2,22}]\\
        \le&\sqrt{\E[\{n(\s_{i_1}^\top\s_{j_1})^2-v_{\gamma,n}\}^2I_{i_1j_1,22}]\cdot \E[\{n(\s_{i_2}^\top\s_{j_2})^2-v_{\gamma,n}\}^2I_{i_2j_2,22}]}=O(p^{-6a}).
    \end{align*}
    Hence,
    \begin{align*}
        \E(W_{i_1j_1}W_{i_2j_2})=o(p^{-2}),
    \end{align*}
    which, joint with \eqref{eq27}, implies
    \begin{align}\label{eq28}
        \E\left(\left|\frac{1}{\sqrt{p^*}}\sum_{1\le i<j\le p}W_{ij}\right|^2\right)=\frac{1}{p^*}\sum_{1\le i<j\le p}\E(W_{ij}^2)+\frac{1}{p^*}\sum_{\substack{(i_1,j_1)\ne(i_1,j_2)\\1\le i_1<j_1\le p\\1\le i_2<j_2\le p}}\E(W_{i_1j_1}W_{i_2j_2})=o(1).
    \end{align}
    This together with \eqref{eq26} readily establishes the claim \eqref{eq22}.
\end{proof}

\subsection{Proof of Theorem \ref{asymptotic independence} }\label{appendix6}
\begin{proof}
    By \eqref{eq20} and \eqref{eq21}, we have 
    \begin{align*}
        \frac{T_{\lceil\gamma p^*\rceil}-\mu_{\gamma,n,p}}{\sqrt{p^*\sigma_{\gamma}^2}}=\frac{1}{\sqrt{p^*\sigma_{\gamma}^2}}\sum_{1\le i<j\le p}[h_{v_{\gamma,n}}(\sqrt{n}\s_i^\top\s_j)-\E\{h_{v_{\gamma,n}}(\sqrt{n}\s_i^\top\s_j)\}]+o_p(1).
    \end{align*}
    where $h_x(u):=(u^2-x)\indicate(u^2\ge x)$. By Lemma \ref{lemma11}, it suffices to show that
    \begin{align*}
        n\hat{\rho}_{(p^*+1-s)}~~\text{and}~~\frac{1}{\sqrt{p^*\sigma_{\gamma}^2}}\sum_{1\le i<j\le p}[h_{v_{\gamma,n}}(\sqrt{n}\s_i^\top\s_j)-\E\{h_{v_{\gamma,n}}(\sqrt{n}\s_i^\top\s_j)\}]
    \end{align*}
    are asymptotically independent.
    
    Define 
    \begin{align*}
        A_p\equiv A_p(x):=\left\{\frac{1}{\sqrt{p^*\sigma_{\gamma}^2}}\sum_{1\le i<j\le p}[h_{v_{\gamma,n}}(\sqrt{n}\s_i^\top\s_j)-\E\{h_{v_{\gamma,n}}(\sqrt{n}\s_i^\top\s_j)\}]\le x\right\}.
    \end{align*}
    Recalling the definition of $I_k=(i_k,j_k)\in\Lambda_p,1\le k\le p^*$ in \eqref{eq1}, by \eqref{eq19}, define $B_k\equiv B_k(y):=\{n\hat{\rho}_{I_k}^2>b_p+y\}=\{n(\s_{i_k}^\top\s_{j_k})^2>b_p+y\}$. According to the proof of Theorem \ref{asymptotic distribution of ordered statistics} and \ref{asymptotic normality}, we have $\Pr(A_p)\to \Phi(x)$ and 
    \begin{align*}
	    \Pr\left(\bigcup_{1\le k_1<k_2<\cdots<k_s\le p^*}B_{k_1}\cdots B_{k_s}\right)\to 1-\Lambda(y)\sum_{i=0}^{s-1}\frac{\{\log \Lambda^{-1}(y)\}^i}{i!}.
	\end{align*}
	Our goal is to prove that
	\begin{align*}
	    \Pr\left(\bigcup_{1\le k_1<k_2<\cdots<k_s\le p^*}A_pB_{k_1}\cdots B_{k_s}\right)\to \Phi(x)\left[1-\Lambda(y)\sum_{i=0}^{s-1}\frac{\{\log \Lambda^{-1}(y)\}^i}{i!}\right].
	\end{align*}

    By the Bonferroni inequality given in Lemma \ref{lemma12}, we observe that for any integer $1\le d<[(p^*-s)/2]$,
	\begin{align*}
	\sum_{l=s}^{2d+s+1}(-1)^{s+l}\binom{l-1}{s-1}\sum_{1\le k_1<\cdots<k_l\le p^*}&\Pr(A_pB_{k_1}\cdots B_{k_l})\le \Pr\left(\bigcup_{1\le k_1<k_2<\cdots<k_s\le p^*}A_pB_{k_1}\cdots B_{k_s}\right)\\
	\le& \sum_{l=s}^{2d+s}(-1)^{s+l}\binom{l-1}{s-1}\sum_{1\le k_1<\cdots<k_l\le p^*}\Pr(A_pB_{k_1}\cdots B_{k_l}).
	\end{align*}
	Using the Bonferroni inequality again, we have
	\begin{align*}
	&\sum_{l=s}^{2d+s}(-1)^{s+l}\binom{l-1}{s-1}\sum_{1\le k_1<\cdots<k_l\le p^*}\Pr(A_pB_{k_1}\cdots B_{k_l})\\
	=&\sum_{l=s}^{2d+s-1}(-1)^{s+l}\binom{l-1}{s-1}\sum_{1\le k_1<\cdots<k_l\le p^*}\Pr(A_p)\Pr(B_{k_1}\cdots B_{k_l})\\
	&+\sum_{l=s}^{2d+s-1}(-1)^{s+l}\binom{l-1}{s-1}\sum_{1\le k_1<\cdots<k_l\le p^*}\{\Pr(A_pB_{k_1}\cdots B_{k_l})-\Pr(A_p)\Pr(B_{k_1}\cdots B_{k_l})\}\\
	&+\binom{2d+s-1}{s-1}\sum_{1\le k_1<\cdots<k_{2d+s}\le p^*}\Pr(A_pB_{k_1}\cdots B_{k_{2d+s}})\\
	\le &\Pr(A_p)\Pr\left(\bigcup_{1\le k_1<k_2<\cdots<k_s\le p^*}B_{i_1}\cdots B_{k_s}\right)+\sum_{l=s}^{2d+s-1}\binom{l-1}{s-1}\zeta(p,l)+\binom{2d+s-1}{s-1}H(p,2d+s),
	\end{align*}
    and 
    \begin{align*}
        &\sum_{l=s}^{2d+s+1}(-1)^{s+l}\binom{l-1}{s-1}\sum_{1\le k_1<\cdots<k_l\le p^*}\Pr(A_pB_{k_1}\cdots B_{k_l})\\
        =&\sum_{l=s}^{2d+s+2}(-1)^{s+l}\binom{l-1}{s-1}\sum_{1\le k_1<\cdots<k_l\le p^*}\Pr(A_p)\Pr(B_{k_1}\cdots B_{k_l})\\
	    &+\sum_{l=s}^{2d+s+2}(-1)^{s+l}\binom{l-1}{s-1}\sum_{1\le k_1<\cdots<k_l\le p^*}\{\Pr(A_pB_{k_1}\cdots B_{k_l})-\Pr(A_p)\Pr(B_{k_1}\cdots B_{k_l})\}\\
	    &-\binom{2d+s+1}{s-1}\sum_{1\le k_1<\cdots<k_{2d+s+2}\le p^*}\Pr(A_pB_{k_1}\cdots B_{k_{2d+s}})\\
        \ge&\Pr(A_p)\Pr\left(\bigcup_{1\le k_1<k_2<\cdots<k_s\le p^*}B_{i_1}\cdots B_{k_s}\right)-\sum_{l=s}^{2d+s+2}\binom{l-1}{s-1}\zeta(p,l)-\binom{2d+s+1}{s-1}H(p,2d+s+2),
    \end{align*}
    where 
    \begin{align*}
        \zeta(p,l):=&\sum_{1\le k_1<\cdots<k_l\le p^*}|\Pr(A_pB_{k_1}\cdots B_{k_l})-\Pr(A_p)\Pr(B_{k_1}\cdots B_{k_l})|,~~\text{and}\\
        H(p,l):=&\sum_{1\le k_1<\cdots<k_l\le p^*}\Pr(B_{k_1}\cdots B_{k_l}).
    \end{align*}
    By Lemma \ref{lemma13}, under the assumption that $\log p=o(\sqrt{n})$, we have 
    \begin{align}\label{eq29}
        \limsup_{p\to\infty}H(p,l)=O\left(\frac{1}{l!}\right).
    \end{align}
    If we claim that 
    \begin{align}\label{eq30}
        \limsup_{p\to\infty}\zeta(p,l)=o\left(\frac{1}{l!}\right).
    \end{align}
    Then, we obtain
    \begin{align*}
	    &\limsup_{p\to\infty}\Pr\left(\bigcup_{1\le k_1<k_2<\cdots<k_s\le p^*}A_pB_{k_1}\cdots B_{k_s}\right)\le \Phi(x)\left[1-\Lambda(y)\sum_{i=0}^{s-1}\frac{\{\log \Lambda^{-1}(y)\}^i}{i!}\right],~~\text{and}\\
        &\liminf_{p\to\infty}\Pr\left(\bigcup_{1\le k_1<k_2<\cdots<k_s\le p^*}A_pB_{k_1}\cdots B_{k_s}\right)\ge \Phi(x)\left[1-\Lambda(y)\sum_{i=0}^{s-1}\frac{\{\log \Lambda^{-1}(y)\}^i}{i!}\right],
	\end{align*}
    which leads to the desired result.

    To complete the proof, we prove the calim \eqref{eq30}. For $1\le k_1<k_2<\cdots<k_l\le p^*$, write $I_{k_m}=(i_{k_m},j_{k_m})$ for $1\le m\le l$. Set
    \begin{align*}
        \Lambda_{l,p}:=\{(i_{k_m},j):i_{k_m}<j\le p,1\le m\le l\}\bigcup\{(i,j_{k_m}):1\le i<j_{k_m},1\le m\le l\}
    \end{align*}
    for $l\ge 1$. It is easy to check that $|\Lambda_{l,p}|=\sum_{m=1}^l(p-i_{k_m}+j_{k_m}-2)$. Since $i_{k_m}<j_{k_m}$ for each $m$, we see that 
    \begin{align*}
        l(p-1)\le|\Lambda_{l,p}|\le\sum_{m=1}^l(p+j_{k_m})\le 2lp.
    \end{align*}
    Observe that $B_{k_1}\cdots B_{k_l}$ is an event generated by random vectors $\{\s_i,s_j:(i,j)\in\Lambda_{l,p}\}$. A crucial observation is that
    \begin{align}\label{eq31}
        \sum_{(i,j)\in\Lambda_p-\Lambda_{l,p}}[h_{v_{\gamma,n}}(\sqrt{n}\s_i^\top\s_j)-\E\{h_{v_{\gamma,n}}(\sqrt{n}\s_i^\top\s_j)\}]~~\text{and}~~B_{k_1}\cdots B_{k_l}~~\text{are independent}.
    \end{align}
    It is easy to see that
    \begin{align}\label{eq32}
        &\sum_{(i,j)\in\Lambda_{l,p}}[h_{v_{\gamma,n}}(\sqrt{n}\s_i^\top\s_j)-\E\{h_{v_{\gamma,n}}(\sqrt{n}\s_i^\top\s_j)\}]\nonumber\\
        =&\sum_{m=1}^l\sum_{j=i_{k_m}+1}^p[h_{v_{\gamma,n}}(\sqrt{n}\s_{i_{k_m}}^\top\s_j)-\E\{h_{v_{\gamma,n}}(\sqrt{n}\s_{i_{k_m}}^\top\s_j)\}]\nonumber\\
        &+\sum_{m=1}^l\sum_{i=1}^{j_{k_m}-1}[h_{v_{\gamma,n}}(\sqrt{n}\s_i^\top\s_{j_{k_m}})-\E\{h_{v_{\gamma,n}}(\sqrt{n}\s_i^\top\s_{j_{k_m}})\}]\nonumber\\
        &-\sum_{m=1}^l\sum_{m'=1}^l[h_{v_{\gamma,n}}(\sqrt{n}\s_{i_{k_m}}^\top\s_{j_{k_{m'}}})-\E\{h_{v_{\gamma,n}}(\sqrt{n}\s_{i_{k_m}}^\top\s_{j_{k_{m'}}})\}]\nonumber\\
        =&:Q_{p,1}+Q_{p,2}+Q_{p,3}.
    \end{align}
    For any integer $\tau\ge 2$, by H\"{o}lder inequality that $(\sum_{m=1}^l|a_m|)^\tau\le l^{\tau-1}\sum_{m=1}^l|a_m|^\tau $, we have
    \begin{align*}
        \E(|Q_{p,1}|^{\tau})=&\E\left(\left|\sum_{m=1}^l\sum_{j=i_{k_m}+1}^p[h_{v_{\gamma,n}}(\sqrt{n}\s_{i_{k_m}}^\top\s_j)-\E\{h_{v_{\gamma,n}}(\sqrt{n}\s_{i_{k_m}}^\top\s_j)\}]\right|^\tau\right)\\
        \le&l^{\tau-1}\sum_{m=1}^l\E\left(\left|\sum_{j=i_{k_m}+1}^p[h_{v_{\gamma,n}}(\sqrt{n}\s_{i_{k_m}}^\top\s_j)-\E\{h_{v_{\gamma,n}}(\sqrt{n}\s_{i_{k_m}}^\top\s_j)\}]\right|^\tau\right).
    \end{align*}
    By \eqref{eql3} that $\{\s_1^\top\s_2,\s_1^\top\s_3,\ldots,\s_1^\top\s_p\}$ are independent and Lemma \ref{lemma14}, we have
    \begin{align*}
        &\E\left(\left|\sum_{j=i+1}^p[h_{v_{\gamma,n}}(\sqrt{n}\s_i^\top\s_j)-\E\{h_{v_{\gamma,n}}(\sqrt{n}\s_i^\top\s_j)\}]\right|^\tau\right)\\
        \le&C(p-i)^{\frac{\tau}{2}-1}\sum_{j=i+1}^p\E[|h_{v_{\gamma,n}}(\sqrt{n}\s_i^\top\s_j)-\E\{h_{v_{\gamma,n}}(\sqrt{n}\s_i^\top\s_j)\}|^\tau],
    \end{align*}
    where the constant $C>0$ is free of $n$ and $p$. Note that $h_x(u)=h_{x,1}(u)-h_{x,2}(u)$, where $h_{x,1}(u):=u^2\indicate(u^2\ge x)$ and $h_{x,2}(u):=x\indicate(u^2\ge x)$. Then, by H\"{o}lder inequality again, we have
    \begin{align*}
        &\E[|h_{v_{\gamma,n}}(\sqrt{n}\s_i^\top\s_j)-\E\{h_{v_{\gamma,n}}(\sqrt{n}\s_i^\top\s_j)\}|^\tau]\\
        \le&2^{\tau-1}\sum_{m=1}^2\E[|h_{m,v_{\gamma,n}}(\sqrt{n}\s_i^\top\s_j)-\E\{h_{m,v_{\gamma,n}}(\sqrt{n}\s_i^\top\s_j)\}|^\tau].
    \end{align*}
    Recall \eqref{eql1} that $(\s_i^\top\s_j)^2\sim\text{Beta}\left(\frac{1}{2},\frac{n-2}{2}\right)$ for any $ 1\le i<j\le p$. Define $\xi_1:=\sqrt{n}\s_1^\top\s_2\indicate(|\s_1^\top\s_2|\ge\sqrt{v_{\gamma,n}/n})$, then we can calculate $\E(\xi_1)=0$ and 
    \begin{align*}
        \E(|\xi_1|^{2\tau})=\frac{n^\tau\Gamma\left(\tau+\frac{1}{2}\right)\Gamma\left(\frac{n-1}{2}\right)}{\Gamma\left(\frac{1}{2}\right)\Gamma\left(\tau+\frac{n-1}{2}\right)}\left\{1-F_{\text{Beta}\left(\tau+\frac{1}{2},\frac{n-2}{2}\right)}\left(\frac{v_{\gamma,n}}{n}\right)\right\}.
    \end{align*}
    By Lemma \ref{lemma15} (\romannumeral1), we obtain
    \begin{align*}
        &\E[|h_{1,v_{\gamma,n}}(\sqrt{n}\s_i^\top\s_j)-\E\{h_{1,v_{\gamma,n}}(\sqrt{n}\s_i^\top\s_j)\}|^\tau]\\
        =&\E\{|\xi_1^2-E(\xi_1^2)|^\tau\}\\
        \le&32^\tau\frac{n^\tau\Gamma\left(\tau+\frac{1}{2}\right)\Gamma\left(\frac{n-1}{2}\right)}{\Gamma\left(\frac{1}{2}\right)\Gamma\left(\tau+\frac{n-1}{2}\right)}\left\{1-F_{\text{Beta}\left(\tau+\frac{1}{2},\frac{n-2}{2}\right)}\left(\frac{v_{\gamma,n}}{n}\right)\right\}\\
        \approx&32^\tau\frac{\Gamma\left(\tau+\frac{1}{2}\right)}{\Gamma\left(\frac{1}{2}\right)}\{1-F_{\chi^2_{2\tau+1}}(v_{\gamma})\}=O(1),
    \end{align*}
    where the approximation follows from the Stirling expansion and the fact that $F_{\text{Beta}\left(a,\frac{n-2}{2}\right)}\left(\frac{x}{n}\right)\to F_{\text{Gamma}(a,1)}\left(\frac{x}{2}\right)=F_{\chi^2_{2a}}(x)$ and $v_{\gamma,n}\to v_{\gamma}$ as $n\to\infty$.

    Define $\xi_2:=\indicate(|\s_1^\top\s_2|\ge\sqrt{v_{\gamma,n}/n})$, then we can calculate $\E(\xi_2)=1-F_{\text{Beta}\left(\frac{1}{2},\frac{n-2}{2}\right)}\left(\frac{v_{\gamma,n}}{n}\right)=:a_{\gamma,n}$ and 
    \begin{align*}
        \E\{|\xi_2-\E(\xi_2)|^{2\tau}\}=&a_{\gamma,n}(1-a_{\gamma,n})^{2\tau}+a_{\gamma,n}^{2\tau}(1-a_{\gamma,n}).
    \end{align*}
    By Lemma \ref{lemma15} (\romannumeral2), we obtain
    \begin{align*}
        &\E[|h_{2,v_{\gamma,n}}(\sqrt{n}\s_i^\top\s_j)-\E\{h_{2,v_{\gamma,n}}(\sqrt{n}\s_i^\top\s_j)\}|^\tau]\\
        =&v_{\gamma,n}^\tau\E\{|\xi_2^2-E(\xi_2^2)|^\tau\}\\
        \le&v_{\gamma,n}^\tau\cdot 16^\tau\left\{(1-a_{\gamma,n})^\tau+\sqrt{(1-a_{\gamma,n})^{2\tau}a_{\gamma,n}+(1-a_{\gamma,n})a_{\gamma,n}^{2\tau}}\right\}\\
        &\quad\times\left\{a_{\gamma,n}^\tau+(1-a_{\gamma,n})^\tau+\sqrt{(1-a_{\gamma,n})^{2\tau}a_{\gamma,n}+(1-a_{\gamma,n})a_{\gamma,n}^{2\tau}}\right\}=O(1),
    \end{align*}
    due to $a_{\gamma,n}\to 1-F_{\chi^2_1}(v_{\gamma})$ as $n\to\infty$. Therefore,
    \begin{align}\label{eq33}
        \E(|Q_{p,1}|^{\tau})\le Cl^\tau p^{\tau/2},
    \end{align}
    where the constant $C>0$ only depends on $\tau$. Similarly, we can prove
    \begin{align}\label{eq34}
        \E(|Q_{p,2}|^{\tau})\le Cl^\tau p^{\tau/2},
    \end{align}
    and
    \begin{align}\label{eq35}
        \E(|Q_{p,3}|^{\tau})\le l^{2(\tau-1)}\sum_{m=1}^l\sum_{m'=1}^l\E[|h_{v_{\gamma,n}}(\sqrt{n}\s_{i_{k_m}}^\top\s_{j_{k_{m'}}})-\E\{h_{v_{\gamma,n}}(\sqrt{n}\s_{i_{k_m}}^\top\s_{j_{k_{m'}}})\}|^\tau] \le Cl^{2\tau}.
    \end{align}
    Combining with \eqref{eq32}-\eqref{eq35}, we obtain
    \begin{align*}
        \E\left(\left|\sum_{(i,j)\in\Lambda_{l,p}}[h_{v_{\gamma,n}}(\sqrt{n}\s_i^\top\s_j)-\E\{h_{v_{\gamma,n}}(\sqrt{n}\s_i^\top\s_j)\}]\right|^\tau\right)\le C(l^\tau p^{\tau/2}+l^{2\tau})\le C'l^\tau p^{\tau/2},
    \end{align*}
    which, by Markov inequality, implies for any fixed $\varepsilon\in(0,1)$,
    \begin{align}\label{eq36}
        \Pr\left(\frac{1}{\sqrt{p^*\sigma^2_\gamma}}\left|\sum_{(i,j)\in\Lambda_{l,p}}[h_{v_{\gamma,n}}(\sqrt{n}\s_i^\top\s_j)-\E\{h_{v_{\gamma,n}}(\sqrt{n}\s_i^\top\s_j)\}]\right|\ge\varepsilon\right)\le Cl^\tau\varepsilon^{-\tau}p^{-\tau/2}.
    \end{align}

    Recalling the definition of $A_p(x)$, by \eqref{eq31} and \eqref{eq36}, we have
    \begin{align*}
        &\Pr(A_p(x)B_{k_1}\cdots B_{k_l})\\
        \le&\Pr\left(A_p(x)B_{k_1}\cdots B_{k_l},\frac{1}{\sqrt{p^*\sigma^2_\gamma}}\left|\sum_{(i,j)\in\Lambda_{l,p}}[h_{v_{\gamma,n}}(\sqrt{n}\s_i^\top\s_j)-\E\{h_{v_{\gamma,n}}(\sqrt{n}\s_i^\top\s_j)\}]\right|<\varepsilon\right)+Cl^\tau\varepsilon^{-\tau}p^{-\tau/2}\\
        \le&\Pr\left(\frac{1}{\sqrt{p^*\sigma^2_\gamma}}\sum_{(i,j)\in\Lambda_p-\Lambda_{l,p}}[h_{v_{\gamma,n}}(\sqrt{n}\s_i^\top\s_j)-\E\{h_{v_{\gamma,n}}(\sqrt{n}\s_i^\top\s_j)\}]\le x+\varepsilon,B_{k_1}\cdots B_{k_l}\right)+Cl^\tau\varepsilon^{-\tau}p^{-\tau/2}\\
        =&\Pr\left(\frac{1}{\sqrt{p^*\sigma^2_\gamma}}\sum_{(i,j)\in\Lambda_p-\Lambda_{l,p}}[h_{v_{\gamma,n}}(\sqrt{n}\s_i^\top\s_j)-\E\{h_{v_{\gamma,n}}(\sqrt{n}\s_i^\top\s_j)\}]\le x+\varepsilon\right)\Pr(B_{k_1}\cdots B_{k_l})+Cl^\tau\varepsilon^{-\tau}p^{-\tau/2}.
    \end{align*}
    To proceed,
    \begin{align*}
        &\Pr\left(\frac{1}{\sqrt{p^*\sigma^2_\gamma}}\sum_{(i,j)\in\Lambda_p-\Lambda_{l,p}}[h_{v_{\gamma,n}}(\sqrt{n}\s_i^\top\s_j)-\E\{h_{v_{\gamma,n}}(\sqrt{n}\s_i^\top\s_j)\}]\le x+\varepsilon\right)\\
        \le&\Pr\left(\frac{1}{\sqrt{p^*\sigma^2_\gamma}}\sum_{(i,j)\in\Lambda_p-\Lambda_{l,p}}[h_{v_{\gamma,n}}(\sqrt{n}\s_i^\top\s_j)-\E\{h_{v_{\gamma,n}}(\sqrt{n}\s_i^\top\s_j)\}]\le x+\varepsilon,\right.\\
        &\qquad \left.\frac{1}{\sqrt{p^*\sigma^2_\gamma}}\left|\sum_{(i,j)\in\Lambda_{l,p}}[h_{v_{\gamma,n}}(\sqrt{n}\s_i^\top\s_j)-\E\{h_{v_{\gamma,n}}(\sqrt{n}\s_i^\top\s_j)\}]\right|<\varepsilon\right)+Cl^\tau\varepsilon^{-\tau}p^{-\tau/2}\\
        \le&\Pr\left(\frac{1}{\sqrt{p^*\sigma^2_\gamma}}\sum_{(i,j)\in\Lambda_p}[h_{v_{\gamma,n}}(\sqrt{n}\s_i^\top\s_j)-\E\{h_{v_{\gamma,n}}(\sqrt{n}\s_i^\top\s_j)\}]\le x+2\varepsilon\right)+Cl^\tau\varepsilon^{-\tau}p^{-\tau/2}\\
        =&\Pr(A_p(x+2\varepsilon))+Cl^\tau\varepsilon^{-\tau}p^{-\tau/2}.
    \end{align*}
    Hence,
    \begin{align*}
        \Pr(A_p(x)B_{k_1}\cdots B_{k_l})\le \Pr(A_p(x+2\varepsilon))\Pr(B_{k_1}\cdots B_{k_l})+2Cl^\tau\varepsilon^{-\tau}p^{-\tau/2}.
    \end{align*}
    Similarly, we can prove
    \begin{align*}
        \Pr(A_p(x)B_{k_1}\cdots B_{k_l})\ge \Pr(A_p(x-2\varepsilon))\Pr(B_{k_1}\cdots B_{k_l})-2Cl^\tau\varepsilon^{-\tau}p^{-\tau/2}.
    \end{align*}
    Thus,
    \begin{align}
        |\Pr(A_p(x)B_{k_1}\cdots B_{k_l})-\Pr(A_p(x))\Pr(B_{k_1}\cdots B_{k_l})|\le\Delta_{p,\varepsilon}\Pr(B_{k_1}\cdots B_{k_l})+4Cl^\tau\varepsilon^{-\tau}p^{-\tau/2},
    \end{align}
    where
    \begin{align*}
        \Delta_{p,\varepsilon}:=&|\Pr(A_p(x))-\Pr(A_p(x+2\varepsilon))|+|\Pr(A_p(x))-\Pr(A_p(x-2\varepsilon))|\\
        \to&|\Phi(x)-\Phi(x+2\varepsilon)|+|\Phi(x)-\Phi(x-2\varepsilon)|~~\text{as}~~p\to\infty.
    \end{align*}
    This, by the definitions of $\zeta(p,l)$ and $H(p,l)$, implies
    \begin{align*}
        \zeta(p,l)\le \Delta_{p,\varepsilon}H(p,l)+C\binom{p^*}{l}l^\tau\varepsilon^{-\tau}p^{-\tau/2}.
    \end{align*}
    Then, by \eqref{eq29} and the fact that $\binom{p^*}{l}l^\tau p^{-\tau/2}\le p^{2l}l^\tau p^{-\tau/2}\le l^\tau p^{-l}$ when $\tau=6l$, we have
    \begin{align}
        \limsup_{p\to\infty}\zeta(p,l)\le \frac{C}{l!}|\Phi(x)-\Phi(x+2\varepsilon)|+|\Phi(x)-\Phi(x-2\varepsilon)|=o\left(\frac{1}{l!}\right),
    \end{align}
    by setting $\varepsilon\downarrow 0$. The proof of \eqref{eq30} is completed. Thus, we prove (\romannumeral1). (\romannumeral2) could be easily obtained by (\romannumeral1) and Theorem \ref{asymptotic distribution of ordered statistics}.
\end{proof}

\bibliographystyle{chicago}
\bibliography{ref}

@book{Feller1968Probability,
    author    = {William Feller},
    publisher = {Wiley},
    title     = {An introduction to probability theory and its applications},
    year      = {1968},
}

@book{durrett2019probability,
  title={Probability: Theory and Examples},
  author={Durrett, R. and Durrett, R.},
  year={2019},
  publisher={Cambridge University Press}
}

@book{Leadbetter1983ExtremesAR,
    author    = {Leadbetter, M. Ross and Lindgren, Georg and Rootzen, Holger},
    publisher = {Springer},
    title     = {Extremes and related properties of random sequences and processes},
    year      = {1983}
}

@book{Muirhead1982AspectsOM,
    author={Robb J. Muirhead},
    title={Aspects of Multivariate Statistical Theory},
    publisher={Wiley Series in Probability and Statistics},
    year={1982}
}

@article{feng2024asymptotic,
    author  = {Feng, Long and Jiang, Tiefeng and Li, Xiaoyun and Liu, Binghui},
    title   = {Asymptotic independence of the sum and maximum of dependent random variables with applications to high-dimensional tests},
    journal = {Statistica Sinica},
    year    = {2024},
    volume  = {34},
    pages   = {1745--1763},
}

@article{AGG89poisson_approximation,
author = {R. Arratia and L. Goldstein and L. Gordon},
title = {Two Moments Suffice for Poisson Approximations: The Chen-Stein Method},
volume = {17},
journal = {The Annals of Probability},
number = {1},
pages = {9 -- 25},
year = {1989},
}

@article{kendall1938new,
  title={A new measure of rank correlation},
  author={Kendall, Maurice},
  journal={Biometrika},
  volume={30},
  number={1-2},
  pages={81--93},
  year={1938},
  publisher={Oxford University Press}
}

@article{hoeffding1948non,
  title={A non-parametric test of independence},
  author={Hoeffding, Wassily},
  journal={The Annals of Mathematical Statistics},
  volume={19},
  number={4},
  pages={546--557},
  year={1948},
  publisher={JSTOR}
}

@article{pearson1895vii,
  title={VII. Note on regression and inheritance in the case of two parents},
  author={Pearson, Karl},
  journal={proceedings of the royal society of London},
  volume={58},
  number={347-352},
  pages={240--242},
  year={1895},
  publisher={The Royal Society London}
}

@article{li2015joint,
  title={Joint limiting laws for high-dimensional independence tests},
  author={Li, Danning and Xue, Lingzhou},
  journal={arXiv preprint arXiv:1512.08819},
  year={2015}
}

@article{ma2024adaptive,
  title={Adaptive L-statistics for high dimensional test problem},
  author={Ma, Huifang and Feng, Long and Wang, Zhaojun},
  journal={arXiv preprint arXiv:2410.14308},
  year={2024}
}

@article{xia2025consistent,
  title={Consistent complete independence test in high dimensions based on Chatterjee correlation coefficient},
  author={Xia, Liqi and Cao, Ruiyuan and Du, Jiang and Dai, Jun},
  journal={Statistical Papers},
  volume={66},
  number={1},
  pages={3},
  year={2025},
  publisher={Springer}
}

@article{drton2020high,
  title={High-dimensional consistent independence testing with maxima of rank correlations},
  author={Drton, Mathias and Han, Fang and Shi, Hongjian},
  journal={The Annals of Statistics},
  volume={48},
  number={6},
  pages={3206--3227},
  year={2020},
  publisher={JSTOR}
}

@article{yao2018testing,
  title={Testing mutual independence in high dimension via distance covariance},
  author={Yao, Shun and Zhang, Xianyang and Shao, Xiaofeng},
  journal={Journal of the Royal Statistical Society Series B: Statistical Methodology},
  volume={80},
  number={3},
  pages={455--480},
  year={2018},
  publisher={Oxford University Press}
}

@article{han2017distribution,
  title={Distribution-free tests of independence in high dimensions},
  author={Han, Fang and Chen, Shizhe and Liu, Han},
  journal={Biometrika},
  volume={104},
  number={4},
  pages={813--828},
  year={2017},
  publisher={Oxford University Press}
}

@article{cai2024asymptotic,
  title={Asymptotic distribution-free independence test for high-dimension data},
  author={Cai, Zhanrui and Lei, Jing and Roeder, Kathryn},
  journal={Journal of the American Statistical Association},
  volume={119},
  number={547},
  pages={1794--1804},
  year={2024},
  publisher={Taylor \& Francis}
}

@article{liu2020cauchy,
  title={Cauchy combination test: a powerful test with analytic p-value calculation under arbitrary dependency structures},
  author={Liu, Yaowu and Xie, Jun},
  journal={Journal of the American Statistical Association},
  volume={115},
  number={529},
  pages={393--402},
  year={2020}
}

@book{serfling1980approximation,
  title={Approximation Theorems of Mathematical Statistics},
  author={Serfling, Robert J.},
  edition={1st},
  publisher={Wiley - Interscience},
  address={New York},
  year={1980},
  series={Wiley Series in Probability and Mathematical Statistics: Probability and Mathematical Statistics},
  isbn={0471024031}
}

@article{li2023test,
  title={A test for the identity of a high-dimensional correlation matrix based on the l4-norm},
  author={Li, Weiming and Xiong, Xiaoge},
  journal={Journal of Statistical Planning and Inference},
  volume={225},
  pages={132--145},
  year={2023},
  publisher={Elsevier}
}

@book{anderson2003introduction,
  title={An Introduction to Multivariate Statistical Analysis},
  author={Anderson, T.W.},
  edition={3rd},
  publisher={John Wiley \& Sons},
  address={Hoboken, NJ},
  year={2003}
}

@article{bai2004clt,
  title={CLT for linear spectral statistics of large-dimensional sample covariance matrices},
  author={Bai, Z. and Silverstein, J.W.},
  journal={Annals of Probability},
  volume={32},
  pages={553--605},
  year={2004}
}

@article{cai2011limiting,
  title={Limiting laws of coherence of random matrices with applications to testing covariance structure and construction of compressed sensing matrices},
  author={Cai, T.T. and Jiang, T.},
  journal={Annals of Statistics},
  volume={39},
  number={3},
  pages={1496--1525},
  year={2011}
}

@article{heiny2024log,
  title={Log determinant of large correlation matrices under infinite fourth moment},
  author={Heiny, Johannes and Parolya, Nestor},
  journal={Annales de l'Institut Henri Poincare (B) Probabilites et statistiques},
  volume={60},
  number={2},
  pages={1048--1076},
  year={2024},
  organization={Institut Henri Poincar{\'e}}
}

@article{hu2023limiting,
  title={Limiting distributions of the likelihood ratio test statistics for independence of normal random vectors},
  author={Hu, Mingyue and Qi, Yongcheng},
  journal={Statistical Papers},
  volume={64},
  number={3},
  pages={923--954},
  year={2023},
  publisher={Springer}
}

@article{gao2017high,
  title={High dimensional correlation matrices: the central limit theorem and its applications},
  author={Gao, J. and Han, X. and Pan, G. and Yang, Y.},
  journal={Journal of the Royal Statistical Society. Series B (Statistical Methodology)},
  volume={79},
  number={3},
  pages={677--693},
  year={2017}
}

@article{jiang2004asymptotic,
  title={The asymptotic distributions of the largest entries of sample correlation matrices},
  author={Jiang, T.},
  journal={Annals of Applied Probability},
  volume={14},
  number={2},
  pages={865--880},
  year={2004}
}

@article{jiang2019determinant,
  title={Determinant of sample correlation matrix with application},
  author={Jiang, T.},
  journal={Annals of Applied Probability},
  volume={29},
  number={3},
  pages={1356--1397},
  year={2019}
}

@article{jiang2013central,
  title={Central limit theorems for classical likelihood ratio tests for high-dimensional normal distributions},
  author={Jiang, T. and Yang, F.},
  journal={Annals of Statistics},
  volume={41},
  number={4},
  pages={2029--2074},
  year={2013}
}

@article{liu2008asymptotic,
  title={The asymptotic distribution and Berry--Esseen bound of a new test for independence in high dimension with an application to stochastic optimization},
  author={Liu, W. and Lin, Z. and Shao, Q.},
  journal={Annals of Applied Probability},
  volume={18},
  number={6},
  pages={2337--2366},
  year={2008}
}

@article{schott2005testing,
  title={Testing for complete independence in high dimensions},
  author={Schott, J.R.},
  journal={Biometrika},
  volume={92},
  number={4},
  pages={951--956},
  year={2005}
}

@article{srivastava2005some,
  title={Some tests concerning the covariance matrix in high dimensional data},
  author={Srivastava, M.S.},
  journal={Journal of the Japan Statistical Society},
  volume={35},
  number={2},
  pages={251--272},
  year={2005}
}

@article{srivastava2006some,
  title={Some tests criteria for the covariance matrix with fewer observations than the dimension},
  author={Srivastava, M.S.},
  journal={Acta Commentationes Universitatis Tartuensis de Mathematica},
  volume={10},
  pages={77--93},
  year={2006}
}

@article{srivastava2011some,
  title={Some tests for the covariance matrix with fewer observations than the dimension under non-normality},
  author={Srivastava, M.S. and Kollo, T. and von Rosen, D.},
  journal={Journal of Multivariate Analysis},
  volume={102},
  number={6},
  pages={1090--1103},
  year={2011}
}

@article{wang2024rank,
  title={Rank-based max-sum tests for mutual independence of high-dimensional random vectors},
  author={Wang, Hongfei and Liu, Binghui and Feng, Long and Ma, Yanyuan},
  journal={Journal of Econometrics},
  volume={238},
  number={1},
  pages={105578},
  year={2024},
  publisher={Elsevier}
}

@article{feng2022max,
  title={Max-sum tests for cross-sectional independence of high-dimensional panel data},
  author={Feng, Long and Jiang, Tiefeng and Liu, Binghui and Xiong, Wei},
  journal={The Annals of Statistics},
  volume={50},
  number={2},
  pages={1124--1143},
  year={2022},
  publisher={Institute of Mathematical Statistics}
}

@article{spearman1987proof,
  title={The proof and measurement of association between two things},
  author={Spearman, Charles},
  journal={The American journal of psychology},
  volume={100},
  number={3/4},
  pages={441--471},
  year={1987},
  publisher={JSTOR}
}

\end{document}